\documentclass[preprint,12pt]{elsarticle}
\usepackage{amsmath,amsthm,amssymb}
\usepackage{booktabs}
\usepackage{changepage}
\usepackage{float}
\usepackage{graphicx}
\usepackage[colorlinks=true,linkcolor=black,citecolor=black,urlcolor=blue]{hyperref}
\hypersetup{
   pdftitle={Compounded Linear Failure Rate Distribution: Properties, Simulation and Analysis},
   pdfauthor={Suchismita Das; Akul Ameya; Cahyani Karunia Putri}
}
\usepackage{bookmark}
\usepackage{placeins}
\usepackage{ragged2e}
\usepackage{setspace}
\usepackage{siunitx}


\allowdisplaybreaks
\newtheorem{theorem}{Theorem}[section]

\newtheorem{remark}[theorem]{Remark}
\DeclareGraphicsExtensions{.pdf,.png,.jpg}
\newcommand{\orcid}[1]{%
  \textsuperscript{\href{https://orcid.org/#1}{\raisebox{0pt}{\includegraphics[height=1em]{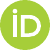}}}}%
}

\biboptions{numbers,sort&compress}

\begin{document}

\begin{frontmatter}

  \title{Compounded Linear Failure Rate Distribution: Properties, Simulation and Analysis}

  \author[spj]{Suchismita Das\corref{cor1}\orcid{0000-0002-9380-6235}}
  \ead{suchismita.das@spjain.org}
  \author[spj]{Akul Ameya\orcid{0009-0003-5050-7934}}
  \author[spj]{Cahyani Karunia Putri\orcid{0009-0002-1015-7186}}

  \cortext[cor1]{Corresponding author.}
  \address[spj]{S P Jain School of Global Management, 15 Carter Street, Sydney NSW 2141, Australia}

  \begin{abstract}
    This paper proposes a new extension of the linear failure rate (LFR) model to better capture real-world lifetime data. The model incorporates an additional shape parameter to increase flexibility. It helps model the minimum survival time from a set of LFR distributed variables. We define the model, derive certain statistical properties such as the mean residual life, the mean inactivity time, moments, quantile, order statistics and also discuss the results on stochastic orders of the proposed distribution. The proposed model has increasing, bathtub shaped and inverse bathtub shaped hazard rate function. We use the method of maximum likelihood estimation to estimate the unknown parameters. We conduct simulation studies to examine the behavior of the estimators. We also use three real datasets to evaluate the model, which turns out superior compared to classical alternatives.
  \end{abstract}


\end{frontmatter}

\onehalfspacing

\section{Introduction}\label{sec_intro}

In reliability engineering, complex systems often experience multiple, independent sources of stress that contribute to their overall degradation and eventual failure.
Modeling such systems requires a probabilistic framework capable of representing how individual stress factors evolve over time and interact to determine the system's lifetime.
Traditionally, lifetime analysis has employed classical probability distributions such as the exponential, Rayleigh and Weibull distributions along with their various generalizations.
However, the growing need to represent more intricate failure behavior has driven significant advancement in the generalization of lifetime distributions, leading to their widespread
applications across diverse fields including engineering, finance, economics and biomedical science \cite{Sen1995, Gupta1999, Mudholkar1993, Ahmad2020}.

Within this context, the Linear Failure Rate Distribution (LFRD) \cite{Bain1974}, also known as the Linear Exponential Distribution (LED) \cite{Sen1995}, has emerged as a pivotal
framework due to its capacity to model both constant and increasing failure rates through a unified structure. The cumulative distribution function (CDF) of the LFR distribution is defined as

\begin{equation*}
  G_0(x) = 1 - \exp\left(-\alpha x - \frac{\beta}{2} x^2\right), \quad x \geq 0, \quad \alpha, \beta >0,
\end{equation*}

where $\alpha$ and $\beta$ are scale parameters. The corresponding hazard function is given by

\begin{equation*}
  h_0(x) = \alpha + \beta x, \quad x \geq 0, \quad \alpha, \beta >0.
\end{equation*}

The quadratic exponential structure enables the LFRD to represent systems with progressive
wear-out mechanisms. Notably, it generalized two classical distributions: when $\beta = 0$,
it reduces to the Exponential Distribution (ED), and when $\alpha = 0$, it reduces to the Rayleigh Distribution (RD).
Despite its versatility, real-world applications often exhibit failure patterns beyond the LFRD's capacity, such as early-life failures or bathtub-shaped hazards. Recent methodological innovations have expanded the model's flexibility through several grouped strategies.

Direct parameter extensions include the Generalized LFRD and Generalized LED \cite{Sarhan2009,Mahmoud2010}, primarily introduced to increase shape flexibility. Generator-based transformations cover the Beta LFRD, Beta GLED, Kumaraswamy GLFRD, Kumaraswamy LED, and McDonald GLFRD \cite{Jafari2015,Shakhatreh2016,Elbatal2013a,Merovci2015,Elbatal2014}, which enhance distributional skewness and tail behaviour. Exponentiation and transmutation methods such as the Exponentiated GLED, Transmuted GLED, Transmuted LED, Extended GLFRD, and Exponentiated Kumaraswamy LED \cite{Sarhan2013,Elbatal2013b,Ghosh2021,Tian2014b,Kazemi2016,Nasiru2018} aim to introduce additional hazard-rate shapes. Composition and compounding approaches, including the Extended LFRD, Marshall--Olkin GLED, Poisson LFRD, and Geometric GLFRD \cite{Ghitany2007,Okasha2016,Cordeiro2015,Nadarajah2012}, are designed to model heterogeneity and unobserved mixing. Further modified developments, such as Modified GLFRD variants, new GLED formulations, Odd Generalized Exponential LFRD, Inverted GLED, Truncated Cauchy Power--G LFRD, and Modified GLED \cite{Jamkhaneh2014,Khan2016,Tian2014a,Poonia2021,Mustafa2016,Mahmoud2017,Alsadat2023,Mahmoud2023}, focus on improving tail control and hazard adaptability.

In this paper, we introduce a new three parameter lifetime distribution as a new modified form of the LFR model. Chakraborty et al. \cite{Chakraborty2025} proposed a Poisson compounded minimum lifetime model using a heavy tailed Burr distribution as the baseline for complex network applications.
However, such models are less suitable for reliability data, where aging effects and hazard rate behavior play an important role.
Motivated by this, we propose a compounded linear failure rate distribution to model survival data. The physical rationale underlying the proposed model is described below with a motivating example.

Suppose an aircraft wing spar is subjected to $N > 0$ independent stress sources like turbulence, engine vibration, and thermal stress. Each of these sources can independently cause microcracks to grow and ultimately lead to component failure. Let $X_1, X_2, \ldots, X_N$ represent the potential times to failure associated with $N$ latent shocks. The number of latent shocks beyond the first is assumed to follow a Poisson distribution with mean $\lambda > 0$, which ensures that at least one shock is always present, i.e, $N - 1 \sim \text{Poisson}(\lambda)$. Each $X_i$ is an independent and identical LFRD variable since the chance of failure initially is small, and characterised by the parameter $\alpha$, which represents the inherent failure rate when the system is new. However, as time progresses and stress exposure increases, the failure rate rises almost linearly over time $(\beta x)$. Accordingly, the actual failure time of the wing spar can be modelled as
$U = \min\{X_1, X_2, \ldots, X_N\}$, implying that the component fails when any one of the independent stress mechanisms reaches its critical failure threshold.

Given this framework, a probabilistic model of this ‘first failure’ approach is recommended to assess component performance under multiple stress factors. By combining a Poisson-distributed number of risks with LFR-distributed latent times, this study proposes the Compounded LFRD (CLFRD) model, which captures variability in failure mechanisms, offering a flexible representation of real-world lifetime data. Consequently, the observed lifetime, given by $U = \min\{X_1, X_2, \ldots, X_N\}$, follows a new family of distributions, reflecting the first failure mechanism in such a system. This model enables accurate lifetime prediction and the development of optimised maintenance and inspection schedules that reflect the underlying failure probabilities.

The rest of the paper is organized as follows. We discuss our new proposed CLFRD model in \autoref{sec_model}, followed by its properties in \autoref{sec_prop}. \autoref{sec_est} deals with parameter estimation, while \autoref{sec_sim} deals with simulation. \autoref{sec_data} showcases empirical performance and finally, \autoref{sec_con} presents some concluding remarks.

\section{The Compounded Linear Failure Rate Distribution}\label{sec_model}

Let $X_1, X_2, \ldots, X_N$ be a random sample from the LFRD and $N - 1 \sim \text{Poisson}(\lambda)$, where $N \in \mathbb{N}$ and we define $U = \min\{X_1, X_2, \ldots, X_N\}$.
The survival function for $U$ is given by:
\begin{align}
  \bar{G}(x) & = P(U \geq x) = \sum_{n=1}^{\infty} P(U \geq x \mid N-1=n-1) P(N-1=n-1) \nonumber                                                          \\
             & = \sum_{n=1}^{\infty} \left[\exp\left(-\alpha x - \frac{\beta}{2}x^2\right)\right]^n \frac{\exp(-\lambda) \lambda^{n-1}}{(n-1)!} \nonumber \\
             & = \exp\left[-\alpha x - \frac{\beta}{2}x^2 - \lambda + \lambda \exp\left(-\alpha x - \frac{\beta}{2}x^2\right)\right], \nonumber           \\
             & \hspace{2em} x \ge 0,\quad \alpha,\beta,\lambda>0 .
\end{align}

The cumulative distribution function (CDF) is therefore given by:
\begin{align}
  G(x) & = 1 - \exp\left[-\alpha x - \frac{\beta}{2} x^2 - \lambda + \lambda \exp\left(-\alpha x - \frac{\beta}{2} x^2\right)\right],\nonumber \\
       & \hspace{2em} x \ge 0,\quad \alpha,\beta,\lambda>0 .
\end{align}
The corresponding probability density function (PDF) of CLFRD$(\alpha, \beta, \lambda)$ is given as:
\begin{equation}\label{eq1}
  g(x) = \frac{(\alpha + \beta x)\left[1 + \lambda \exp\left(-\alpha x - \frac{\beta}{2} x^2\right)\right]}{\exp\left[\alpha x + \frac{\beta}{2} x^2 + \lambda - \lambda \exp\left(-\alpha x - \frac{\beta}{2} x^2\right)\right]}, \quad x \geq 0, \quad \alpha, \beta, \lambda > 0.
\end{equation}

The following result shows certain conditions under which the PDF is unimodal and is decreasing.

\begin{theorem}
  Let $U$ be a random variable having CLFRD$(\alpha, \beta, \lambda)$. The PDF of $U$ is:
  \begin{itemize}
    \item[(i)] Unimodal (with mode at zero) when \( \alpha^2 < \frac{\beta(1+\lambda)}{\lambda+(1+\lambda)^2} \).
    \item[(ii)] Decreasing when \( \alpha^2 \geq \frac{\beta(1+\lambda)}{\lambda+(1+\lambda)^2} \).
  \end{itemize}
\end{theorem}

\begin{proof}
  Differentiating $g(x)$ with respect to $x$ yields:
  \begin{multline*}
    g'(x) = \Biggl[\beta\left(1+\lambda \exp\left(-\alpha x - \frac{\beta}{2} x^2\right)\right) - \lambda(\alpha + \beta x)^2 \exp\left(-\alpha x - \frac{\beta}{2} x^2\right) \\
      - (\alpha + \beta x)^2 (1 + \lambda \exp\left(-\alpha x - \frac{\beta}{2} x^2\right))^2\Biggr] \exp\Biggl(-\alpha x - \frac{\beta}{2} x^2 - \lambda + \\ \lambda \exp\left(-\alpha x - \frac{\beta}{2} x^2\right)\Biggr).
  \end{multline*}
  Substituting \( y = \alpha x + \frac{\beta}{2} x^2 \), we get
  \begin{equation*}
    g'(y) = \theta(y) \exp(-y - \lambda(1-\exp(-y))),
  \end{equation*}
  where
  \begin{equation*}
    \theta(y) = \beta [1 + \lambda \exp(-y)] - (\alpha^2 + 2\beta y) [\lambda \exp(-y) + (1 + \lambda \exp(-y))^2 ].
  \end{equation*}

  Now, $\theta(0) = \beta(1 + \lambda) - \alpha^2(\lambda + (1 + \lambda)^2),$
  or, equivalently, $g'(0) = \beta(1 + \lambda) - \alpha^2(\lambda + (1 + \lambda)^2)$
  and it can be shown that $g'(y)$ is decreasing in $y$.
  Therefore, $g'(y)$ has no zeros or one zero if \( \alpha^2 \geq \frac{\beta(1 + \lambda)}{\lambda + (1 + \lambda)^2} \) or \( \alpha^2 < \frac{\beta(1 + \lambda)}{\lambda + (1 + \lambda)^2} \) on $(0, \infty)$.
  Furthermore, it follows that \( g(0) = \alpha(1 + \lambda) \) and $g(\infty) = 0$ and $g(x)$ is non-negative.
  Therefore, $g(x)$ is unimodal if \( \alpha^2 < \frac{\beta(1 + \lambda)}{\lambda + (1 + \lambda)^2} \), and decreasing if \( \alpha^2 > \frac{\beta(1 + \lambda)}{\lambda + (1 + \lambda)^2} \) for all \( x \geq 0 \).
\end{proof}

\autoref{fig_pdf} showcases the density function plots for the 2 given conditions.

\begin{figure}[htbp]
  \centering
  \begin{minipage}{0.48\textwidth}
    \centering
    \includegraphics[width=0.9\linewidth]{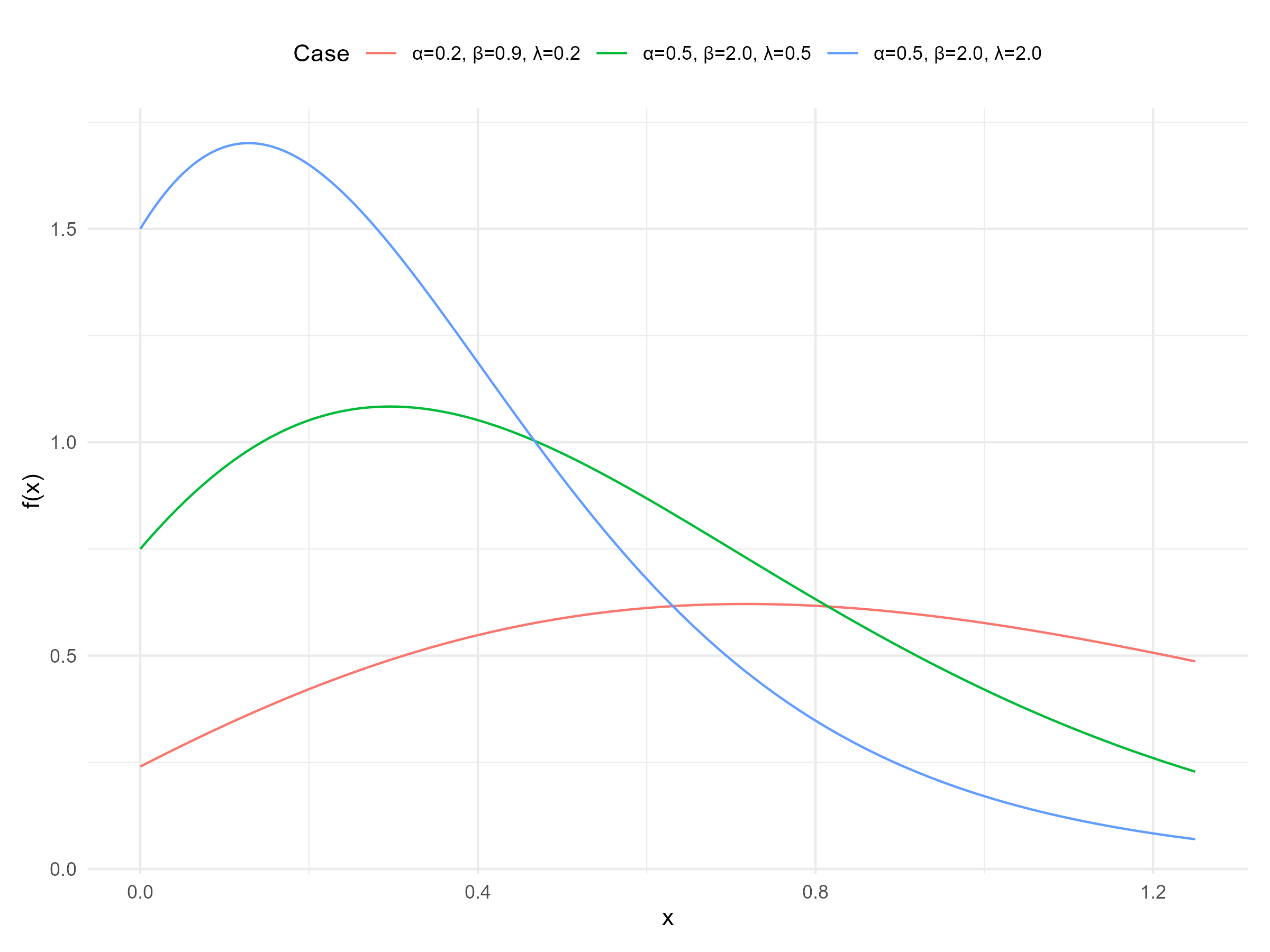}
    {\vspace{0.5ex}\centering\scriptsize unimodal\par}
  \end{minipage}%
  \hspace{0.04\textwidth}%
  \begin{minipage}{0.48\textwidth}
    \centering
    \includegraphics[width=0.9\linewidth]{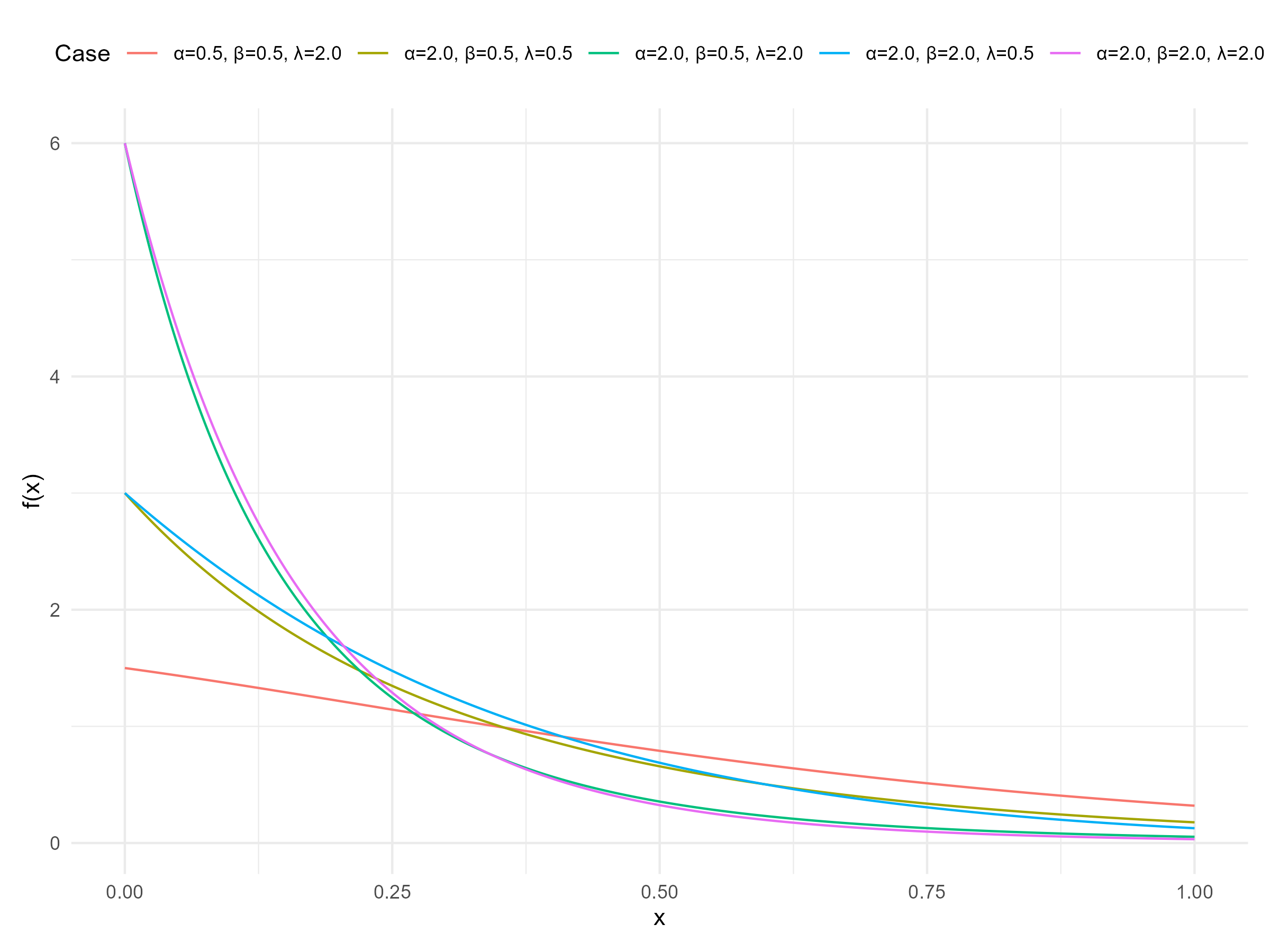}
    {\vspace{0.5ex}\centering\scriptsize increasing\par}
  \end{minipage}
  \caption{Plot corresponding to CLFRD Density Functions for different values of parameters}
  \label{fig_pdf}
\end{figure}

\FloatBarrier

The hazard rate and reversed hazard rate functions are respectfully given for $U \sim \text{CLFRD}(\alpha,\beta,\lambda)$ by:
\begin{equation}
  h(x) = \frac{g(x)}{\bar{G}(x)} = (\alpha + \beta x)\left[1 + \lambda \exp\left(-\alpha x - \frac{\beta}{2} x^2\right)\right], \quad x \geq 0
\end{equation}
and
\begin{equation}
  r(x) = \frac{g(x)}{G(x)} = \frac{(\alpha + \beta x)\left[1 + \lambda \exp\left(-\alpha x - \frac{\beta}{2} x^2\right)\right]}{\exp\left[\alpha x + \frac{\beta}{2} x^2 + \lambda - \lambda \exp\left(-\alpha x - \frac{\beta}{2} x^2\right)\right] - 1}, \quad x \geq 0.
\end{equation}

The following result gives the conditions on the parameters $\alpha, \beta$ and $\lambda$ that describe the possible shapes of $h(x)$ of CLFRD$(\alpha, \beta, \lambda)$.

\begin{theorem}
  Let $U \sim \text{CLFRD}(\alpha, \beta, \lambda)$. Then the hazard rate function $h(x)$ exhibits the following shapes under the specified conditions:

  \begin{itemize}
    \item[(i)] $h(x)$ is monotonically increasing if
          \begin{itemize}
            \item[(a)] \( \alpha^2 < 3\beta \) and \( 0 < \log(2 \lambda) \leq \frac{3 \beta - \alpha^2}{2 \beta} \), or
            \item[(b)] \( \alpha^2 \geq 3\beta \) and \( \beta (1 + \lambda) \geq \lambda \alpha^2 \);
          \end{itemize}

    \item[(ii)] $h(x)$ exhibits a bathtub shape, i.e., decreasing\textendash\hspace{0pt}increasing if
          \begin{itemize}
            \item[(a)] \( \alpha^2 < 3\beta \), \( \beta (1 + \lambda) \leq \lambda \alpha^2 \), and \( \log(2 \lambda) > \frac{3 \beta - \alpha^2}{2\beta} \), or
            \item[(b)] \( \alpha^2 \geq 3\beta \) and \( \beta (1 + \lambda) \leq \lambda \alpha^2 \);
          \end{itemize}

    \item[(iii)] $h(x)$ exhibits an inverse bathtub shape, i.e., increasing\textendash\hspace{0pt}decreasing\textendash\hspace{0pt}increasing if
          \begin{itemize}
            \item[(a)] \( \alpha^2 < 3\beta \) and \( \log(2 \lambda) > \frac{3 \beta - \alpha^2}{2\beta} \).
          \end{itemize}
  \end{itemize}
\end{theorem}

\begin{proof}
  We begin by differentiating the hazard rate function $h(x)$ with respect to $x$:

  \begin{equation*}
    h'(x) = \beta + \lambda [\beta - (\alpha + \beta x)^2] \exp\left(-\alpha x - \frac{\beta}{2} x^2\right).
  \end{equation*}

  Alternatively, making the substitution \( y = \alpha x + \frac{\beta}{2} x^2 \), we obtain:

  \begin{equation*}
    h'(y)= \beta + \lambda [ \beta - \alpha^2 - 2\beta y ] \exp(-y).
  \end{equation*}

  Note that the initial value and limit are given by:
  \begin{equation*}
    h'(0) = \beta + \lambda (\beta - \alpha^2 ) = \beta (1 + \lambda ) - \lambda \alpha^2 \mathpunct{,} \quad \text{and} \quad \lim_{y \to \infty} h'(y) = \beta > 0.
  \end{equation*}

  The second derivative is given by \( h''(y) = \exp(-y) [ 2 \beta y + \alpha^2 - 3 \beta ] \).

  \medskip
  \textbf{Case (i):} $h(x)$ is increasing if \( h'(x) \geq 0 \) for all $x$, or equivalently, \( h'(y) \geq 0 \) for all $y$. This occurs when \( h'(0) \geq 0 \) and $h'(y)$ is monotonically increasing.

  Now, $h'(y)$ is increasing if \( h''(y) \geq 0 \), which holds when:
  \begin{equation*}
    \exp(-y)[2\beta y + \alpha^2 - 3\beta] \geq 0 \mathpunct{,} \quad \text{or equivalently} \mathpunct{,} \quad \alpha^2 \geq 3\beta.
  \end{equation*}

  When $h''(y)$ has a critical point, i.e., $h''(y)=0$, we find:
  \begin{equation*}
    \exp(-y)[2\beta y + \alpha^2 - 3\beta] = 0 \mathpunct{,} \quad \text{or} \quad y_0 = \frac{3\beta - \alpha^2}{2\beta} > 0 \mathpunct{,} \quad \text{provided} \quad \alpha^2 \leq 3\beta.
  \end{equation*}

  At this critical point:
  \begin{align*}
    h'(y_0) & = h'\left(\frac{3 \beta - \alpha^2}{2 \beta}\right) = \beta + \lambda \left[\beta - \alpha^2 - 2 \beta \cdot \frac{3 \beta - \alpha^2}{2 \beta}\right] \exp(-y_0) \\
            & = \beta + \lambda [\beta - \alpha^2 - 3 \beta + \alpha^2] \exp(-y_0) = \beta - 2 \beta \lambda \exp(-y_0).
  \end{align*}

  This expression is non-negative when \( 0 < \log(2 \lambda) \leq \frac{3 \beta - \alpha^2}{2 \beta} \).

  Therefore, $h(x)$ is increasing under:
  \begin{itemize}
    \item[(a)] \( \alpha^2 < 3 \beta \) and \( 0 < \log(2 \lambda) \leq \frac{3 \beta - \alpha^2}{2 \beta} \);
    \item[(b)] \( \alpha^2 \geq 3 \beta \) and \( \beta (1 + \lambda) \geq \lambda\alpha^2 \).
  \end{itemize}

  \medskip
  \textbf{Case (ii):}
  $h(x)$ exhibits a bathtub shape, i.e., decreasing\textendash\hspace{0pt}increasing if $h'(y)$ has exactly one real zero.

  This occurs under either condition:
  \begin{itemize}
    \item[(a)] \( h'(0) \leq 0 \) and $h'(y)$ is monotonically increasing, or
    \item[(b)] \( h'(0) \leq 0 \) and \( h'(y_0) \leq 0 \) (with \( y_0 = \frac{3 \beta - \alpha^2}{2 \beta} \)) while $h'(y)$ is increasing.
  \end{itemize}

  For $y > 0$, the second derivative vanishes when:
  \begin{equation*}
    h''(y) = e^{-y}\,[2\beta y + \alpha^2 - 3\beta] = 0 .
  \end{equation*}

  Therefore,
  \begin{equation*}
    y_0 = \frac{3\beta - \alpha^2}{2\beta} > 0 \mathpunct{,} \qquad \text{provided } \alpha^2 < 3\beta .
  \end{equation*}

  Additionally, we require:
  \begin{equation*}
    h'(0) = \beta (1 + \alpha) - \lambda \alpha^2 \leq 0 \mathpunct{,} \quad \text{and} \quad h'(y_0) = \beta - 2 \lambda \beta \exp(-y_0) \leq 0.
  \end{equation*}

  Thus, when \( \alpha^2 < 3 \beta \), the conditions become:
  \begin{equation*}
    \beta (1 + \alpha) - \lambda \alpha^2 \leq 0 \mathpunct{,} \quad \text{and} \quad \log(2 \lambda) > \frac{3 \beta - \alpha^2}{2 \beta}.
  \end{equation*}

  When \( \alpha^2 \geq 3 \beta \), the condition simplifies to \( \beta (1 + \alpha) - \lambda \alpha^2 \leq 0 \), with $h'(y)$ increasing in $y$.

  \medskip
  \textbf{Case (iii):} $h(x)$ exhibits an inverse bathtub shape, i.e., increasing\textendash\hspace{0pt}decreasing\textendash\hspace{0pt}increasing if \( h'(0) > 0 \) and \( h'(y_0) < 0 \), with $h'(y)$ increasing beyond $y_0$.

  This behavior emerges when:
  \begin{equation*}
    \alpha^2 < 3 \beta, \quad \beta (1 + \lambda) \geq \lambda \alpha^2 \mathpunct{,} \quad \text{and} \quad \log(2 \lambda) > \frac{3 \beta - \alpha^2}{2\beta}.
  \end{equation*}

  The corresponding hazard rate shapes for these three cases are illustrated in \autoref{fig_haz}.
\end{proof}

\FloatBarrier

\begin{figure}[htbp]
  \centering
  \begin{minipage}{0.3\textwidth}
    \centering
    \includegraphics[width=0.9\linewidth]{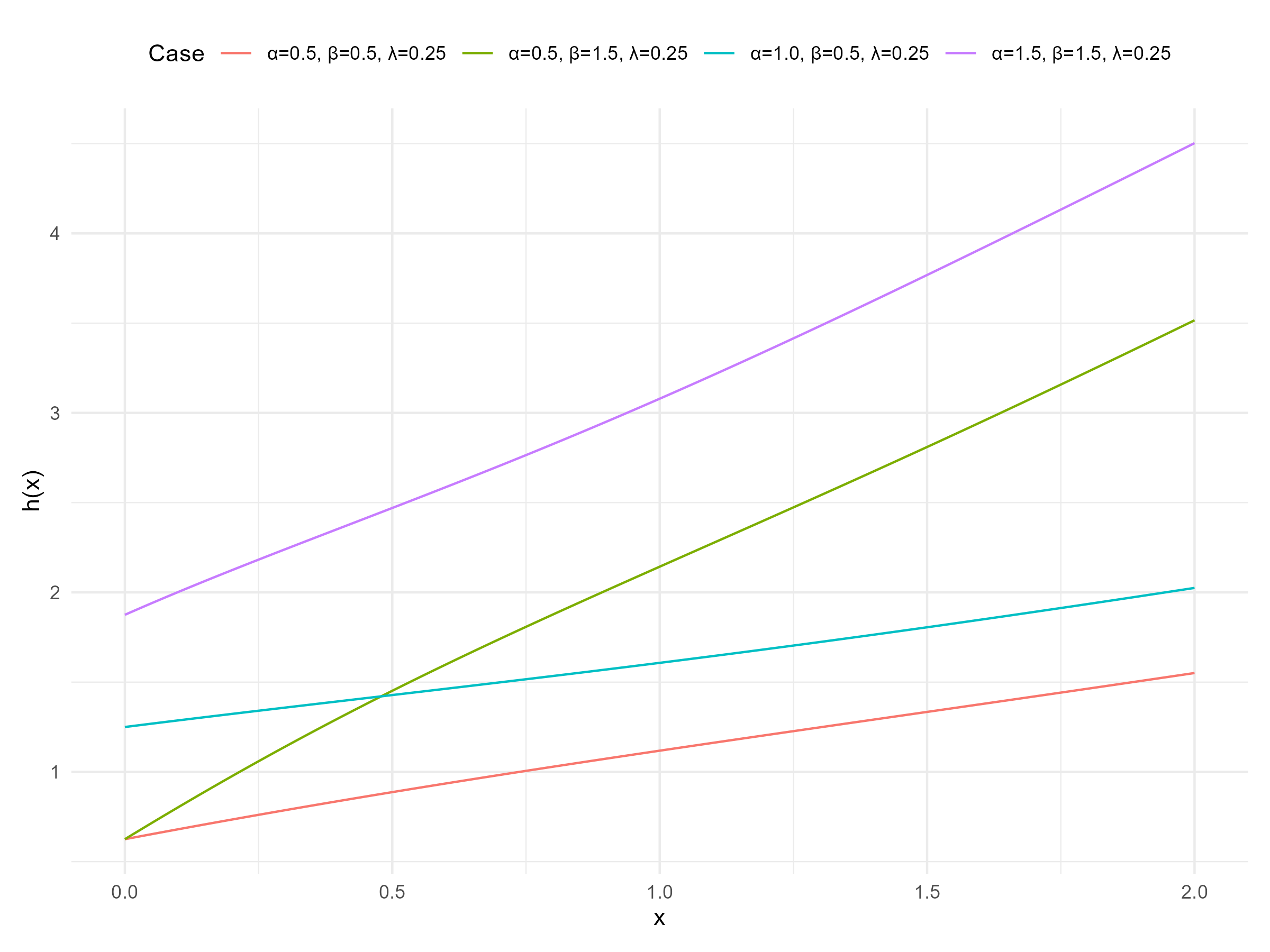}
    {\vspace{0.5ex}\centering\scriptsize increasing\par}
  \end{minipage}%
  \hspace{0.02\textwidth}%
  \begin{minipage}{0.3\textwidth}
    \centering
    \includegraphics[width=0.9\linewidth]{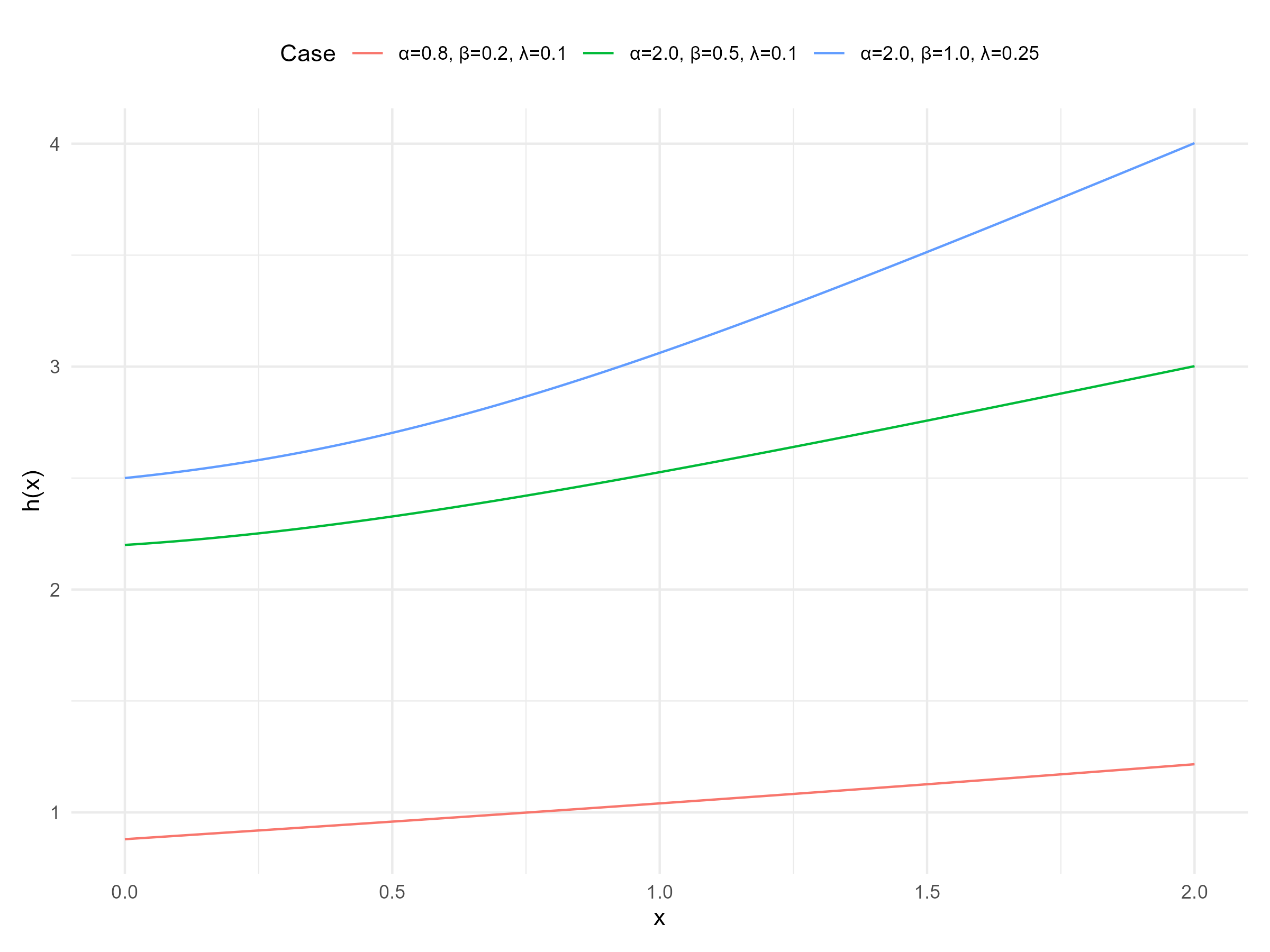}
    {\vspace{0.5ex}\centering\scriptsize increasing\par}
  \end{minipage}%
  \hspace{0.02\textwidth}%
  \begin{minipage}{0.3\textwidth}
    \centering
    \includegraphics[width=0.9\linewidth]{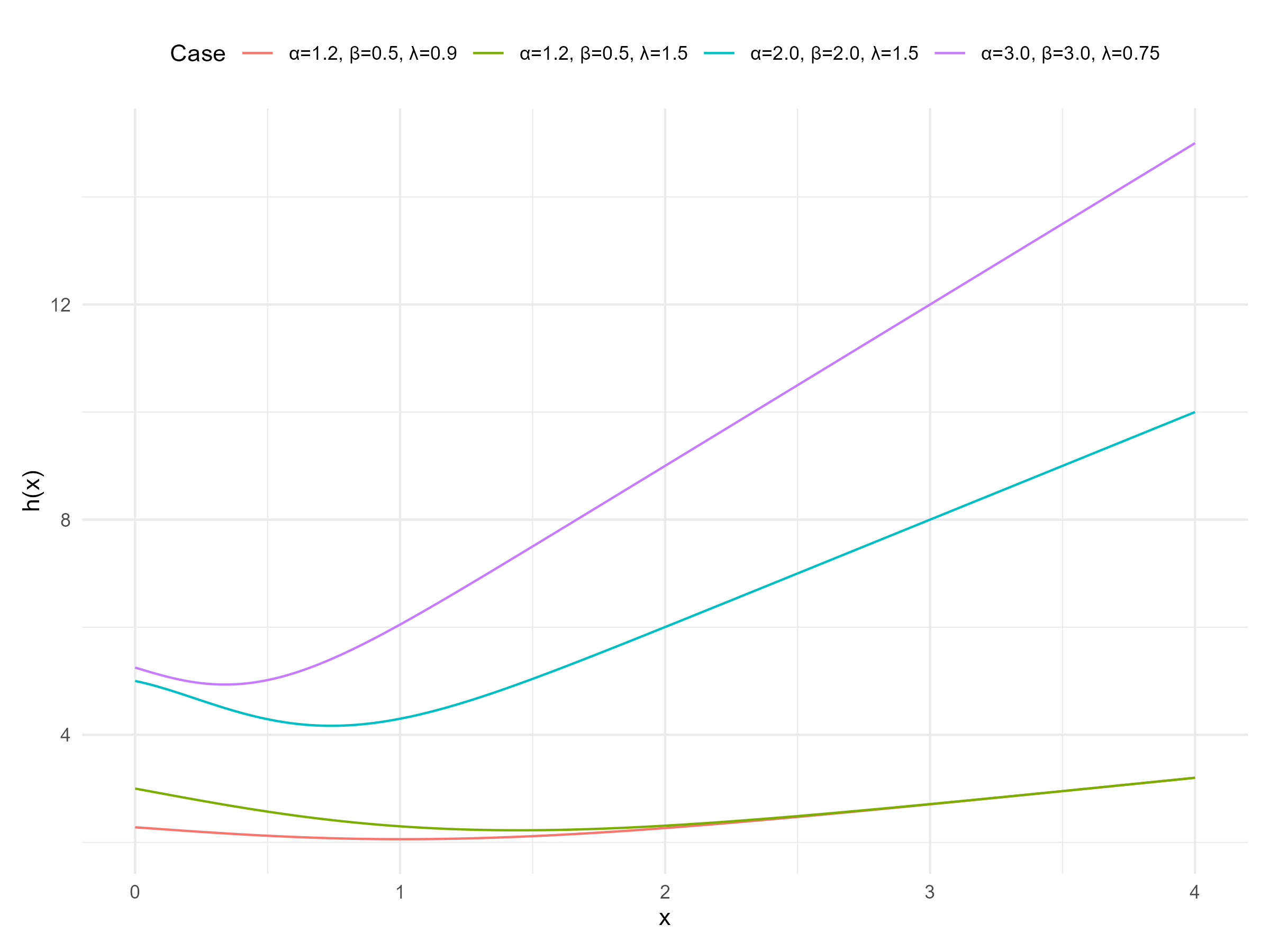}
    {\vspace{0.5ex}\centering\scriptsize bathtub shape\par}
  \end{minipage}

  \vspace{0.5em}

  \begin{minipage}{0.48\textwidth}
    \centering
    \includegraphics[width=0.9\linewidth]{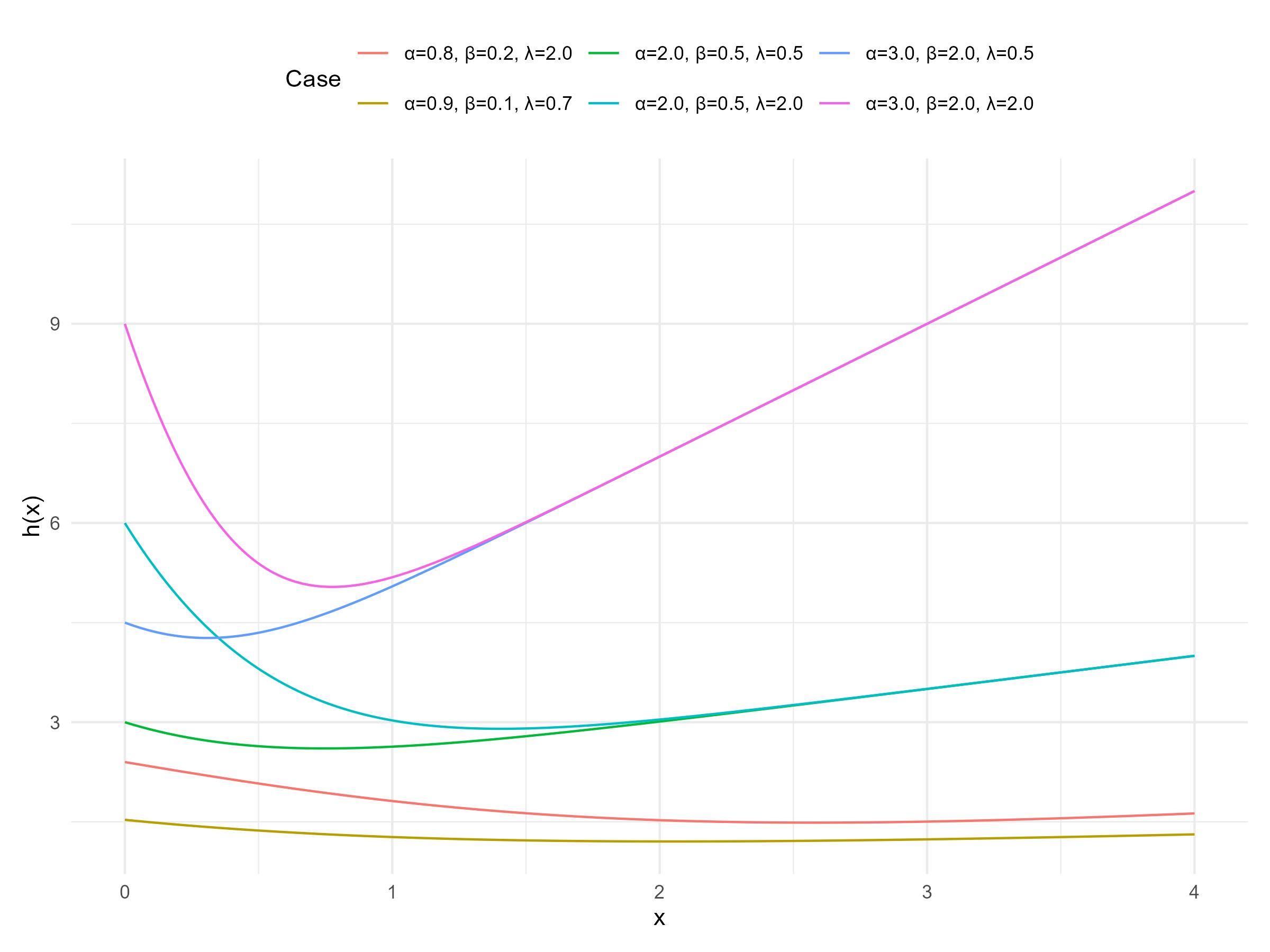}
    {\vspace{0.5ex}\centering\scriptsize bathtub shape\par}
  \end{minipage}%
  \hspace{0.04\textwidth}%
  \begin{minipage}{0.48\textwidth}
    \centering
    \includegraphics[width=0.9\linewidth]{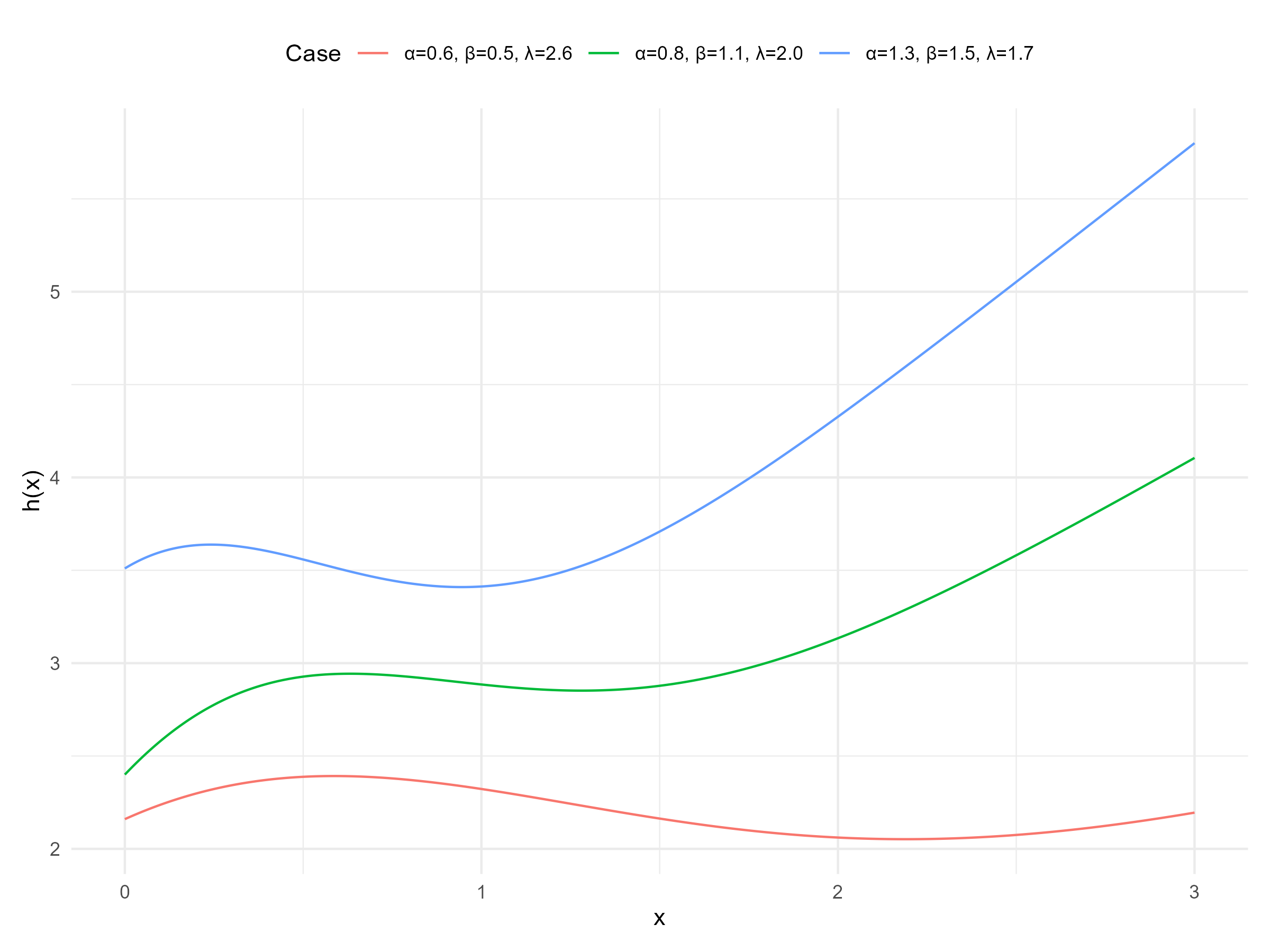}
    {\vspace{0.5ex}\centering\scriptsize inverse bathtub shape\par}
  \end{minipage}
  \caption{Plot of the Hazard Rate functions for different values of parameters}
  \label{fig_haz}
\end{figure}

\FloatBarrier

\section{Statistical Properties}\label{sec_prop}

In probability and statistics, various characterizations of different lifetime
distributions have been proposed. In this section, we have studied some
characteristic statistics such as the mean residual life (MRL), the mean
inactivity time (MIT), moments, and quantiles, along with the results on
stochastic orders and order statistics of the proposed model.

\subsection{Mean Residual Life Function}

The mean residual life (MRL) function is very important in reliability and survival analysis. Consider an aircraft wing spar that is placed under observation. Assume that the wing spar is observed at time $x$ and it is found that it has not failed yet from any of the independent stress sources or shocks. Therefore, the variable of interest is $U_x = (U - x \mid U > x)$, which is known as residual life. It represents the expected remaining lifetime of the component, given that the failure has not occurred till time $x$. The mean residual life of a random variable $U \sim \text{CLFRD}(\alpha, \beta, \lambda)$ is given by:

\begin{equation*}
  M_U(x) = E[U - x \mid U > x] = \frac{1}{\bar{G}(x)} \int_{0}^{\infty} y \, g(y+x) \, dy
\end{equation*}

\begin{multline*}
  = \frac{1}{\exp\left[-\alpha x - \frac{\beta}{2}x^2 - \lambda + \lambda \exp\left(-\alpha x - \frac{\beta}{2}x^2\right)\right]} \Biggl[ \sum_{i=0}^{\infty} \sum_{j=0}^{\infty} \binom{j}{i} \frac{(-1)^{i+j} \lambda^{j}}{j!} \\
    \times \int_{0}^{\infty} y ( \alpha + \beta (y+x) ) \exp\left(-(i+1)\left(\alpha(y+x) + \frac{\beta}{2} (y+x)^2\right)\right) \, dy \\
    + \lambda \sum_{i=0}^{\infty} \sum_{j=0}^{\infty} \binom{j}{i} \frac{(-1)^{i+j} \lambda^{j}}{j!} \int_{0}^{\infty} y ( \alpha + \beta (y+x) ) \\
    \times \exp\left(-(i+2)\left(\alpha(y+x) + \frac{\beta}{2} (y+x)^2\right)\right) \, dy \Biggr].
\end{multline*}
\begin{multline*}
  \text{Now,} \quad \sum_{i=0}^{\infty} \sum_{j=0}^{\infty} \binom{j}{i} \frac{(-1)^{i+j} \lambda^{j}}{j!} \int_{0}^{\infty} y ( \alpha + \beta (y+x) ) \\
  \times \exp\left(-(i+1)\left(\alpha(y+x) + \frac{\beta}{2} (y+x)^2\right)\right) \, dy
\end{multline*}
\begin{multline*}
  = \sum_{i,j,k=0}^{\infty} \binom{j}{i} \lambda^j \beta^k 2^{-k} \frac{(-1)^{i+j+k} (i+1)^k}{j! k!} \int_{0}^{\infty} (y+x-x)(\alpha+\beta(y+x))(y+x)^{2k} \\
  \times \exp(-(i+1)\alpha(y+x)) \, dy
\end{multline*}
\begin{multline*}
  = \sum_{i,j,k=0}^{\infty} b_{ijk}\Biggl[ \alpha \int_{0}^{\infty} (y+x)^{2k+1} \exp(-(i+1)\alpha(y+x)) \, dy \\
    + \beta \int_{0}^{\infty} (y+x)^{2k+2} \exp(-(i+1)\alpha(y+x)) \, dy \\
    - \alpha x \int_{0}^{\infty} (y+x)^{2k} \exp(-(i+1)\alpha(y+x)) \, dy \\
    - \beta x \int_{0}^{\infty} (y+x)^{2k+1} \exp(-(i+1)\alpha(y+x)) \, dy\Biggr]
\end{multline*}
\begin{multline*}
  = \sum_{i,j,k=0}^{\infty} b_{ijk} \Biggl[ \frac{\alpha \Gamma(2k+2)}{[ (i+1)\alpha ]^{2k+2}} + \frac{\beta \Gamma(2k+3)}{[ (i+1)\alpha ]^{2k+3}} \\
    - \frac{\alpha x\Gamma(2k+1)}{[ (i+1)\alpha ]^{2k+1}} - \frac{\beta x\Gamma(2k+2)}{[ (i+1)\alpha ]^{2k+2}} \Biggr].
\end{multline*}

Therefore, the MRL for CLFRD($\alpha, \beta, \lambda$) is:

\begin{multline}
  M_U(x) = \exp\left[\alpha x + \frac{\beta}{2}x^2 + \lambda - \lambda \exp\left(-\alpha x - \frac{\beta}{2}x^2\right)\right] \\
  \times \Biggl[ \sum_{i,j,k=0}^{\infty} b_{ijk} \left( \frac{(\alpha + \beta x)\Gamma(2k+2)}{[(i+1)\alpha]^{2k+2}} - \frac{\alpha x\Gamma(2k+1)}{[(i+1)\alpha]^{2k+1}} + \frac{\beta\Gamma(2k+3)}{[(i+1)\alpha]^{2k+3}} \right) \\
  + \lambda \sum_{i,j,k=0}^{\infty} b'_{ijk}\left( \frac{(\alpha + \beta x)\Gamma(2k+2)}{[(i+2)\alpha]^{2k+2}} - \frac{\alpha x\Gamma(2k+1)}{[(i+2)\alpha]^{2k+1}} + \frac{\beta \Gamma(2k+3)}{[(i+2)\alpha]^{2k+3}} \right) \Biggr],
\end{multline}

where \( b_{ijk} = \binom{j}{i} \lambda^j \beta^k 2^{-k} \frac{(-1)^{i+j+k} (i+1)^k}{j! k!} \) and \( b'_{ijk} = \binom{j}{i} \lambda^j \beta^k 2^{-k} \frac{(-1)^{i+j+k} (i+2)^k}{j! k!} \).

\autoref{table_mrl} shows the MRL value at $x = 0.5$ for 8 different parameter combinations.

\begin{table}[htbp]
  \centering
  \caption{Mean Residual Life at $x = 0.5$}
  \label{table_mrl}
  \vspace{1em}
  \setlength{\tabcolsep}{12pt}
  \begin{tabular}{cccc}
    \toprule
    $\alpha$ & $\beta$ & $\lambda$ & MRL$_{0.5}$ \\
    \midrule
    2.0      & 2.0     & 2.0       & 0.2211234   \\
             &         & 0.5       & 0.2668270   \\
             & 0.5     & 2.0       & 0.2994618   \\
             &         & 0.5       & 0.3773817   \\
    0.5      & 2.0     & 2.0       & 0.2831307   \\
             &         & 0.5       & 0.3962282   \\
             & 0.5     & 2.0       & 0.5214610   \\
             &         & 0.5       & 0.7728661   \\
    \bottomrule
  \end{tabular}
  \vspace{2em}
\end{table}

Larger $\alpha$ values yield a higher MRL, which means that the systems last longer on average after $x = 0.5$.
MRL increases significantly for smaller $\alpha$ and $\beta$ values, which means less failure pressure.

\subsection{Mean Inactivity Time Function}

Consider an aircraft wing spar that is placed under observation.
In most practical situations, continuous monitoring is usually
not possible. Assume that the wing spar is observed at time $x$,
and it is found that it has already failed due to any of the
independent stress sources of shocks. Therefore, the random variable
of interest is $U_x = (x - U \mid U \leq x)$, which is known as inactivity
time \cite{Ruiz1996}. It represents the time that has elapsed since the
component's failure, given that the failure occurred at or before time $x$.
Inactivity time plays a significant role in the development of maintenance
strategies. The mean (expected) inactivity time, also called reversed mean
residual life or expected stopped time \cite{Navarro1997}, of CLFRD $(\alpha, \beta, \lambda)$ is defined as:
\begin{multline*}
  m_{U}(x) = E[x - U \mid U \leq x]= \frac{1}{G(x)} \left[ \int_{0}^{x} (x - u) \, g(u) \, du \right] \\
  = \frac{1}{1 - \exp\left[-\alpha x - \frac{\beta}{2} x^2 - \lambda + \lambda \exp\left(-\alpha x - \frac{\beta}{2} x^2\right)\right]} \\
  \times \Biggl[ \sum_{i=0}^{\infty} \sum_{j=0}^{\infty} \binom{j}{i} \frac{(-1)^{i+j} \lambda^{j}}{j!} \int_{0}^{x} y \, f(y) \, \bar{F}^i(y) \, dy  \\
    + \lambda \sum_{i=0}^{\infty} \sum_{j=0}^{\infty} \binom{j}{i} \frac{(-1)^{i+j} \lambda^{j}}{j!} \int_{0}^{x} y \, f(y) \, \bar{F}^{i+1}(y) \, dy \Biggr].
\end{multline*}

Now,
\begin{multline*}
  \sum_{i=0}^{\infty} \sum_{j=0}^{\infty} \binom{j}{i} \frac{(-1)^{i+j} \lambda^{j}}{j!} \int_{0}^{x} y \, f(y) \, \bar{F}^i(y) \, dy \\
  = \sum_{i,j,k=0}^{\infty} \binom{j}{i} \lambda^j \beta^k 2^{-k}\frac{(-1)^{i+j+k}(i+1)^k y^{2k}}{j! k!}
  \Biggl[ \alpha \int_{0}^{x} y^{2k+1} \exp(-\alpha (i+1) y) \, dy \Biggr. \\
    \Biggl. + \beta \int_{0}^{x} y^{2k+2} \exp(-\alpha (i+1) y) \, dy \Biggr] \\
  = \sum_{i,j,k=0}^{\infty} \binom{j}{i} \lambda^j \beta^k 2^{-k} \frac{(-1)^{i+j+k} (i+1)^k y^{2k}}{j! k!} \Biggl( \frac{\alpha \gamma(2k+2, \alpha (i+1) x)}{[\alpha (i+1)]^{2k+2}} \Biggr. \\
  \Biggl. + \frac{\beta \gamma(2k+3, \alpha (i+1) x)}{[\alpha (i+1)]^{2k+3}} \Biggr).
\end{multline*}

Therefore, the MIT for CLFRD($\alpha, \beta, \lambda$) is:
\begin{multline}
  m_U(x) = x - \frac{\exp\left[\alpha x + \frac{\beta}{2} x^2 + \lambda - \lambda \exp\left(-\alpha x - \frac{\beta}{2} x^2\right)\right]}{\exp\left[\alpha x + \frac{\beta}{2} x^2 + \lambda - \lambda \exp\left(-\alpha x - \frac{\beta}{2} x^2 \right)\right] - 1} \\
  \times \Biggl[ \sum_{i,j,k=0}^{\infty} b_{ijk} \left( \frac{\alpha \gamma(2k+2, \alpha (i+1) x)}{[\alpha (i+1)]^{2k+2}} + \frac{\beta \gamma(2k+3, \alpha (i+1) x)}{[\alpha (i+1)]^{2k+3}} \right) \\
  + \lambda \sum_{i,j,k=0}^{\infty} b'_{ijk} \left( \frac{\alpha \gamma(2k+2, \alpha (i+2) x)}{[\alpha (i+2)]^{2k+2}} + \frac{\beta \gamma(2k+3, \alpha (i+2) x)}{[\alpha (i+2)]^{2k+3}} \right) \Biggr],
\end{multline}

where \( b_{ijk} = \binom{j}{i} \lambda^j \beta^k 2^{-k} \frac{(-1)^{i+j+k} (i+1)^k}{j! k!} \) and \( b'_{ijk} = \binom{j}{i} \lambda^j \beta^k 2^{-k} \frac{(-1)^{i+j+k} (i+2)^k}{j! k!} \).

\autoref{table_mit} shows the MIT value at $x = 0.5$ for 8 different parameter combinations.

\begin{table}[htbp]
  \centering
  \caption{Mean Inactivity Time at $x = 0.5$}
  \label{table_mit}
  \vspace{1em}
  \setlength{\tabcolsep}{12pt}
  \begin{tabular}{cccc}
    \toprule
    $\alpha$ & $\beta$ & $\lambda$ & MIT$_{0.5}$ \\
    \midrule
    2.0      & 2.0     & 2.0       & 0.3592062   \\
             &         & 0.5       & 0.3090133   \\
             & 0.5     & 2.0       & 0.3578331   \\
             &         & 0.5       & 0.3114150   \\
    0.5      & 2.0     & 2.0       & 0.2714062   \\
             &         & 0.5       & 0.2417515   \\
             & 0.5     & 2.0       & 0.2763928   \\
             &         & 0.5       & 0.2556945   \\
    \bottomrule
  \end{tabular}
  \vspace{2em}
\end{table}

\FloatBarrier

As $\alpha$ increases, MIT tends to decrease, which means that the system becomes more prone to failure.

\subsection{Moments}

In statistical analysis and applications, moments play an important role in describing the key features and characteristics of a distribution. Let $U \sim \text{CLFRD}(\alpha, \beta, \lambda)$, then the $r^{th}$ moment of $U$ for $\lambda > 0$ can be written as:

\begin{equation*}
  E[U^r] = \mu^{(r)} = \int_{0}^{\infty} y^r g(y) dy,
\end{equation*}

or, equivalently,

\begin{multline*}
  \mu^{(r)} = \sum_{i=0}^{\infty} \sum_{j=0}^{\infty} \binom{j}{i} \frac{(-1)^{i+j} \lambda^{j}}{j!} \int_{0}^{\infty} y^r \, f(y) \, \bar{F}^{i}(y) \, dy \\
  + \lambda \sum_{i=0}^{\infty} \sum_{j=0}^{\infty} \binom{j}{i} \frac{(-1)^{i+j} \lambda^{j}}{j!} \int_{0}^{\infty} y^r \, f(y) \, \bar{F}^{i+1}(y) \, dy
\end{multline*}

\begin{multline}
  = \sum_{i,j,k=0}^{\infty} \binom{j}{i} \lambda^j \beta^k 2^{-k} \frac{(-1)^{i+j+k}}{j! k!}\Biggl[ \frac{\alpha (i+1)^k \Gamma (r+2k+1)}{[\alpha (i+1)]^{r+2k+1}}\\
    + \frac{\beta (i+1)^k \Gamma (r+2k+2)}{[\alpha (i+1)]^{r+2k+2}} + \frac{\alpha \lambda (i+2)^k \Gamma (r+2k+1)}{[\alpha (i+2)]^{r+2k+1}} \\
    + \frac{\beta \lambda (i+2)^k \Gamma (r+2k+2)}{[\alpha (i+2)]^{r+2k+2}}\Biggr].
\end{multline}

The corresponding first and second moments about the origin are given by:
\begin{multline*}
  E(U) = \sum_{i,j,k=0}^{\infty} c_{ijk}\Biggl[ \frac{\alpha (i+1)^k \Gamma (2k+2)}{[(i+1)\alpha]^{2k+2}} + \frac{\beta (i+1)^k \Gamma (2k+3)}{[(i+1)\alpha]^{2k+3}} \\
    + \frac{\alpha \lambda (i+2)^k \Gamma (2k+2)}{[(i+2)\alpha]^{2k+2}} + \frac{\beta \lambda (i+2)^k \Gamma (2k+3)}{[(i+2)\alpha]^{2k+3}}\Biggr],
\end{multline*}

and
\begin{multline*}
  E(U^2) = \sum_{i,j,k=0}^{\infty} c_{ijk}\Biggl[ \frac{\alpha (i+1)^k \Gamma (2k+3)}{[(i+1)\alpha]^{2k+3}} + \frac{\beta (i+1)^k \Gamma (2k+4)}{[(i+1)\alpha]^{2k+4}} \\
    + \frac{\alpha \lambda (i+2)^k \Gamma (2k+3)}{[(i+2)\alpha]^{2k+3}} + \frac{\beta \lambda (i+2)^k \Gamma (2k+4)}{[(i+2)\alpha]^{2k+4}}\Biggr],
\end{multline*}

where $c_{ijk}=\binom{j}{i} \lambda^j \beta^k 2^{-k} \frac{(-1)^{i+j+k}}{j! k!}$.
\subsection{Quantile and Median}

\begin{theorem}
  Let $U \sim \text{CLFRD}(\alpha, \beta, \lambda)$. Therefore, the $Q^{th}$ quantile is given by
  \begin{equation}\label{eqn1}
    x_Q = \beta^{-1} \left[-\alpha + \sqrt{\alpha^2 + 2\beta[W_0(\lambda(1-Q)\exp(\lambda)) - \lambda - \log(1-Q)]}\right].
  \end{equation}
\end{theorem}

\begin{proof}
  Let $G(x) = Q$; therefore,
  \begin{equation*}
    1 - Q = \exp\left(-\lambda - \alpha x - \frac{\beta}{2} x^2 + \lambda \exp\left(-(\alpha x + \frac{\beta}{2} x^2)\right)\right), \quad 0 \leq Q \leq 1,
  \end{equation*}

  which is equivalent to
  \begin{equation}
    \label{eqn_quant1}
    [\lambda + \log(1 - Q) + y] \exp(y) = \lambda,
  \end{equation}

  where \( y = \left(\alpha x + \frac{\beta}{2} x^2 \right) \). Suppose that
  \begin{equation}
    \label{eqn_quant2}
    z = \lambda + \log(1 - Q) + y \quad \text{or} \quad
    z - \lambda - \log(1 - Q) = y.
  \end{equation}

  Substituting \eqref{eqn_quant2} into \eqref{eqn_quant1} yields: \( z \exp(z - \lambda - \log(1 - Q)) = \lambda\) or \( z \exp(z) = \lambda \exp(\lambda)(1 - Q) \).
  Applying the Lambert W function yields:
  $$ y + \lambda + \log(1 - Q) = W_0[\lambda(1 - Q)\exp(\lambda)],$$
  or,
  $$\quad y - W_0[\lambda(1-Q)\exp(\lambda)] + \lambda + \log(1-Q) = 0.$$

  Therefore, the $Q^{th}$ quantile of CLFRD$(\alpha, \beta, \lambda)$ is

  \begin{equation}\label{eqn_quant3}
    x_Q = \beta^{-1} \left[- \alpha + \sqrt{\alpha^2 + 2\beta[W_0(\lambda(1-Q)\exp(\lambda)) - \lambda - \log(1-Q)]}\right].
  \end{equation}
\end{proof}

By substituting $Q = \frac{1}{2}$ in \eqref{eqn_quant3}, the median of the new model is obtained as:

\begin{equation*}
  x_{MD} = \beta^{-1} \left(-\alpha + \sqrt{\alpha^2 + 2\beta\left[W_0\left(\frac{\lambda \exp(\lambda)}{2}\right) - \lambda - \log\left(\frac{1}{2}\right)\right]} \right).
\end{equation*}

\subsection{Stochastic Orders}

In the area of reliability and survival analysis, various stochastic orderings between two random variables have been considered.
The following theorem gives the conditions under which CLFRD preserves
likelihood ratio ordering. A random variable $X$ is said to be larger than another random variable $Y$ in likelihood
ratio order (written as $X \geq_{LR} Y$) if for all $x \geq 0$, $g_{X}(x)/g_{Y}(x)$ is increasing in $x$.

\begin{theorem}
  Let $X$ and $Y$ be two random variables following CLFRD with parameters
  $\alpha_{1},~\beta_{1},~\lambda_{1}$ and $\alpha_{2},~\beta_{2},~\lambda_{2}$ respectively,
  then $X\geq_{LR}(\leq_{LR}) Y$ provided
  \begin{itemize}
    \item[(i)] $\alpha_{1} \leq(\geq) \alpha_{2}$;
    \item[(ii)] $\beta_{1} \leq(\geq) \beta_{2}$;
    \item[(iii)] $\lambda_{1} \leq(\geq) \lambda_{2}$.
  \end{itemize}
\end{theorem}
\begin{proof}
  $X \geq_{LR} Y$ if and only if, $g_{X_{CLFRD}}(x)/g_{Y_{CLFRD}}(x)$ is increasing in x. Now,
  \begin{multline*}
    \frac{g_{X_{CLFRD}}(x)}{g_{Y_{CLFRD}}(x)} = \\
    \frac{(\alpha_{1}+\beta_{1}x)\left(1+\lambda_{1}\exp\left[-\alpha_{1}x-\frac{\beta_{1}}{2}x^{2}\right]\right)\exp\left[-\lambda_{1}\left(1-\exp\left(-\alpha_{1}x-\frac{\beta_{1}}{2}x^{2}\right)\right)\right]}{(\alpha_{2}+\beta_{2}x)\left(1+\lambda_{2}\exp\left[-\alpha_{2}x-\frac{\beta_{2}}{2}x^{2}\right]\right)\exp\left[-\lambda_{2}\left(1-\exp\left(-\alpha_{2}x-\frac{\beta_{2}}{2}x^{2}\right)\right)\right]} \\
    \times \exp\left[(\alpha_{2}-\alpha_{1})x+\frac{\beta_{2}-\beta_{1}}{2}x^2\right].
  \end{multline*}
  By simple calculations it can be shown that $\frac{g_{X_{CLFRD}}(x)}{g_{Y_{CLFRD}}(x)}$ is increasing in $x$ if ($i$) $\alpha_{1} \leq \alpha_{2}$, ($ii$) $\beta_{1} \leq \beta_{2}$,
  and ($iii$) $\lambda_{1} \leq \lambda_{2}$.
\end{proof}

\begin{remark}
  From the relationship among the various stochastic orders and the LR order, we have the following.
  $X\geq_{*}(\leq_{*}) Y$ provided
  \begin{itemize}
    \item[(i)] $\alpha_{1} \leq(\geq) \alpha_{2}$;
    \item[(ii)] $\beta_{1} \leq(\geq) \beta_{2}$;
    \item[(iii)] $\lambda_{1} \leq(\geq) \lambda_{2}$.
  \end{itemize}
  where $*$ stands for `usual stochastic order', `hazard rate order', `reversed hazard rate order',`convex order', `concave order', `mean residual life order', `harmonic average mean residual life order' and so on \cite{Shaked2007}.
\end{remark}
\subsection{Order Statistics}

Let $U_1, U_2, \cdots, U_n$ be a random sample from CLFRD$(\alpha, \beta, \lambda)$ having order statistics $U_{(1)}, U_{(2)}, \cdots, U_{(n)}$

The PDF of $U_{k:n}$ for the order statistics is given by:
\begin{multline*}
  g_{k:n}(x) = \frac{n!}{(k-1)! (n-k)!} \\
  \times \frac{(\alpha + \beta x)[1 + \lambda \exp(-\alpha x - \frac{\beta}{2} x^2)]}{\exp\left[(n-k+1)\left(\lambda + \alpha x + \frac{\beta}{2} x^2 - \lambda \exp\left(-\alpha x - \frac{\beta}{2} x^2\right)\right)\right]} \\
  = \sum_{j=0}^{k-1}  \binom{n-1}{k-1} \binom{k-1}{j} \frac{(-1)^j n(\alpha + \beta x)}{\exp(n+j+1-k)} \\
  \times \frac{\left(1 + \lambda \exp\left(-\alpha x - \frac{\beta}{2} x^2\right)\right)}{\exp\left(\lambda + \alpha x + \frac{\beta}{2} x^2 - \lambda \exp\left(-\alpha x - \frac{\beta}{2} x^2\right)\right)} .
\end{multline*}

The PDF of the smallest order statistic is:
\begin{equation*}
  g_{1:n}(x) = \frac{n(\alpha + \beta x) \left[ 1 + \lambda \exp\left(-\alpha x - \frac{\beta}{2} x^2\right)\right]}{\exp\left[n \left( \lambda + \alpha x + \frac{\beta}{2} x^2 - \lambda \exp\left(-\alpha x - \frac{\beta}{2} x^2\right)\right)\right]}.
\end{equation*}

The PDF of the largest order statistic is:
\begin{equation*}
  g_{n:n}(x) = \sum_{j=0}^{n-1} (-1)^j n(\alpha + \beta x) (1 + \lambda \exp(-y)) \exp[-(j+1) ( \lambda + y - \lambda \exp(-y) )].
\end{equation*}

\section{Estimation}\label{sec_est}

Suppose that $X_{1}, X_{2}, ..., X_{n}$ is a random sample of size $n$ from CLFRD with parameters $\alpha, \beta, \lambda$. Based on (\ref{eq1}), the log-likelihood function for the parameter vector $\Theta=(\alpha, \beta, \lambda)^{T}$ can be expressed as follows:
\begin{multline}
  \mathcal{L}(\Theta) = - n \lambda - \alpha \sum_{i=1}^{n} x_i - \frac{1}{2} \beta \sum_{i=1}^{n} x_i^2 + \lambda \sum_{i=1}^{n} \exp\Biggl(-\alpha x_i - \frac{\beta}{2} x_i^2\Biggr) \\
  + \sum_{i=1}^{n} \log(\alpha + \beta x_i) + \sum_{i=1}^{n} \log\Biggl[1 + \lambda \exp\Biggl(-\alpha x_i - \frac{\beta}{2} x_i^2\Biggr)\Biggr].
\end{multline}

The following are the normal equations obtained by setting the partial derivatives of the log-likelihood function with respect to the parameters $\alpha, \beta, \lambda$ equal to zero:
\begin{align*}
  \frac{\partial \mathcal{L}}{\partial \alpha}  & =
  \begin{aligned}[t]
     & \sum_{i=1}^n x_i + \lambda \sum_{i=1}^n x_i \exp\left(-\alpha x_i - \frac{\beta}{2} x_i^2\right) - \sum_{i=1}^n \frac{1}{\alpha + \beta x_i}                   \\
     & \quad - \lambda \sum_{i=1}^n \frac{x_i \exp\left(-\alpha x_i - \frac{\beta}{2} x_i^2\right)}{1 + \lambda \exp\left(-\alpha x_i - \frac{\beta}{2} x_i^2\right)}
  \end{aligned}                                 \\
  \frac{\partial \mathcal{L}}{\partial \beta}   & =
  \begin{aligned}[t]
     & -\frac{\lambda}{2} \sum_{i=1}^n x_i^2 + \frac{1}{2} \sum_{i=1}^n x_i^2 \exp\left(-\alpha x_i - \frac{\beta}{2} x_i^2\right) - \sum_{i=1}^n \frac{x_i}{\alpha + \beta x_i}  \\
     & \quad + \frac{\lambda}{2} \sum_{i=1}^n \frac{x_i^2 \exp\left(-\alpha x_i - \frac{\beta}{2} x_i^2\right)}{1 + \lambda \exp\left(-\alpha x_i - \frac{\beta}{2} x_i^2\right)}
  \end{aligned}                     \\
  \frac{\partial \mathcal{L}}{\partial \lambda} & = -n + \sum_{i=1}^n \frac{\exp\left(-\alpha x_i - \frac{\beta}{2} x_i^2\right)}{1 + \lambda \exp\left(-\alpha x_i - \frac{\beta}{2} x_i^2\right)}
\end{align*}

By setting the above partial derivatives to zero and solving them simultaneously, we can obtain the MLEs of the three parameters for the CLFRD with $\alpha$, $\beta$, $\lambda$. Since the closed form solution for these equations are not available, we have used an iterative algorithm to solve these equations numerically.

We further derive the asymptotic distributions for the CLFRD. Then we compute the covariance matrix $I^{-1}$, which is the inverse of the observed information matrix, and is obtained by evaluating the negative second derivatives of the log-likelihood function w.r.t. $\alpha, \beta, \lambda$. Since it is difficult to find an analytical solution, numerical methods can be used to approximate the covariance matrix.

\begin{align*}
  I^{-1}
   & =
  \begin{pmatrix}
    - \Biggl. \dfrac{\partial^2 \mathcal{L}}{\partial \alpha^2} \Biggr|_{\hat{\alpha}, \hat{\beta}, \hat{\lambda}}                &
    - \Biggl. \dfrac{\partial^2 \mathcal{L}}{\partial \alpha \partial \beta} \Biggr|_{\hat{\alpha}, \hat{\beta}, \hat{\lambda}}   &
    - \Biggl. \dfrac{\partial^2 \mathcal{L}}{\partial \alpha \partial \lambda} \Biggr|_{\hat{\alpha}, \hat{\beta}, \hat{\lambda}}   \\
    - \Biggl. \dfrac{\partial^2 \mathcal{L}}{\partial \alpha \partial \beta} \Biggr|_{\hat{\alpha}, \hat{\beta}, \hat{\lambda}}   &
    - \Biggl. \dfrac{\partial^2 \mathcal{L}}{\partial \beta^2} \Biggr|_{\hat{\alpha}, \hat{\beta}, \hat{\lambda}}                 &
    - \Biggl. \dfrac{\partial^2 \mathcal{L}}{\partial \beta \partial \lambda} \Biggr|_{\hat{\alpha}, \hat{\beta}, \hat{\lambda}}    \\
    - \Biggl. \dfrac{\partial^2 \mathcal{L}}{\partial \alpha \partial \lambda} \Biggr|_{\hat{\alpha}, \hat{\beta}, \hat{\lambda}} &
    - \Biggl. \dfrac{\partial^2 \mathcal{L}}{\partial \beta \partial \lambda} \Biggr|_{\hat{\alpha}, \hat{\beta}, \hat{\lambda}}  &
    - \Biggl. \dfrac{\partial^2 \mathcal{L}}{\partial \lambda^2} \Biggr|_{\hat{\alpha}, \hat{\beta}, \hat{\lambda}}
  \end{pmatrix}^{-1} \\
   & =
  \begin{pmatrix}
    \mathrm{Var}(\hat{\alpha})                & \mathrm{Cov}(\hat{\alpha}, \hat{\beta})  & \mathrm{Cov}(\hat{\alpha}, \hat{\lambda}) \\
    \mathrm{Cov}(\hat{\alpha}, \hat{\beta})   & \mathrm{Var}(\hat{\beta})                & \mathrm{Cov}(\hat{\beta}, \hat{\lambda})  \\
    \mathrm{Cov}(\hat{\alpha}, \hat{\lambda}) & \mathrm{Cov}(\hat{\beta}, \hat{\lambda}) & \mathrm{Var}(\hat{\lambda})
  \end{pmatrix}
  .
\end{align*}

The derivatives in $I^{-1}$ are given as follows:

\begin{align*}
  \frac{\partial^2 \mathcal{L}}{\partial \alpha^2}                & =
  \begin{aligned}[t]
     & \lambda \sum_{i=1}^n x_i^2 \exp\left(-\alpha x_i - \frac{\beta}{2} x_i^2\right) - \sum_{i=1}^n \frac{1}{(\alpha + \beta x_i)^2}                                               \\
     & \quad + \sum_{i=1}^n \frac{\lambda x_i^2 \exp\left(-\alpha x_i - \frac{\beta}{2} x_i^2\right)}{\left(\lambda + \exp\left(-\alpha x_i - \frac{\beta}{2} x_i^2\right)\right)^2}
  \end{aligned}                                                                                                                                                       \\
  \frac{\partial^2 \mathcal{L}}{\partial \alpha \partial \beta}   & =
  \begin{aligned}[t]
     & \frac{\lambda}{2} \sum_{i=1}^n x_i^3 \exp\left(-\alpha x_i - \frac{\beta}{2} x_i^2\right) + \sum_{i=1}^n \frac{x_i}{( \alpha + \beta x_i )^2}                                             \\
     & \quad + \frac{\lambda}{2} \sum_{i=1}^n \frac{x_i^3 \exp\left(-\alpha x_i - \frac{\beta}{2} x_i^2\right)}{\left(1 + \lambda \exp\left(-\alpha x_i - \frac{\beta}{2} x_i^2\right)\right)^2}
  \end{aligned}                                                                                                                                           \\
  \frac{\partial^2 \mathcal{L}}{\partial \alpha \partial \lambda} & = -\sum_{i=1}^n x_i \exp\left(-\alpha x_i - \frac{\beta}{2} x_i^2\right) - \sum_{i=1}^n \frac{x_i \exp\left(-\alpha x_i - \frac{\beta}{2} x_i^2\right)}{\left(1 + \lambda \exp\left(-\alpha x_i - \frac{\beta}{2} x_i^2\right)\right)^2}                             \\
  \frac{\partial^2 \mathcal{L}}{\partial \beta^2}                 & =
  \begin{aligned}[t]
     & \frac{\lambda}{4} \sum_{i=1}^n x_i^4 \exp\left(-\alpha x_i - \frac{\beta}{2} x_i^2\right) - \sum_{i=1}^n x_i^2 (\alpha + \beta x_i)^{-2}                                                  \\
     & \quad + \frac{\lambda}{4} \sum_{i=1}^n \frac{x_i^4 \exp\left(-\alpha x_i - \frac{\beta}{2} x_i^2\right)}{\left(1 + \lambda \exp\left(-\alpha x_i - \frac{\beta}{2} x_i^2\right)\right)^2}
  \end{aligned}                                                                                                                                           \\
  \frac{\partial^2 \mathcal{L}}{\partial \beta \partial \lambda}  & = -\frac{1}{2} \sum_{i=1}^n x_i^2 \exp\left(-\alpha x_i - \frac{\beta}{2} x_i^2\right) - \frac{1}{2} \sum_{i=1}^n \frac{x_i^2 \exp\left(-\alpha x_i - \frac{\beta}{2} x_i^2\right)}{\left(1 + \lambda \exp\left(-\alpha x_i - \frac{\beta}{2} x_i^2\right)\right)^2} \\
  \frac{\partial^2 \mathcal{L}}{\partial \lambda^2}               & = -\sum_{i=1}^n \frac{\exp\left(-2(\alpha x_i + \frac{\beta}{2} x_i^2)\right)}{\left(1 + \lambda \exp\left(-\alpha x_i - \frac{\beta}{2} x_i^2\right)\right)^2}
\end{align*}

The asymptotic distribution of MLEs for $\alpha, \beta, \lambda$ can be written as
\begin{equation*}
  \left(\hat{\theta}-\theta\right) \sim N_{3}\left(0, I^{-1}(\hat{\theta})\right) \mathpunct{,} \quad \text{where} \quad \hat{\theta}=\left(\hat{\alpha}, \hat{\beta}, \hat{\lambda}\right)
\end{equation*}

Then we can find the $100(1-p)\%$ confidence intervals for $\alpha, \beta, \lambda$, given by $\hat{\alpha} \pm z_{p/2}\sqrt{Var(\hat{\alpha})}$, $\hat{\beta} \pm z_{p/2}\sqrt{Var(\hat{\beta})}$ and $\hat{\lambda} \pm z_{p/2}\sqrt{Var(\hat{\lambda})}$. Here, $z_{p/2}$ is the upper $100 p$-th percentile of the standard normal distribution.

\section{Simulation Study}\label{sec_sim}

This section presents a simulation study to check the characteristics of ML estimators in the CLFRD. This suggests that the MLE approach provides reliable estimates of the CLFRD. The simulation was repeated 5000 times for each configuration.
We considered eight different parameter sets for $\alpha, \beta, \lambda$, such as $(2.0, 2.0, 2.0)$, $(2.0, 2.0, 0.5)$, $(2.0, 0.5, 2.0)$, $(2.0, 0.5, 0.5)$, $(0.5, 2.0, 2.0)$, $(0.5, 2.0, 0.5)$, $(0.5, 0.5, 2.0)$, and $(0.5, 0.5, 0.5)$, and created datasets with sample sizes $n = 100, 200,$ and $300$, from CLFRD using \eqref{eqn1}. \autoref{table_sim1} to \autoref{table_sim8}
report the mean and the standard deviation (SD) of the estimates for the parameters $\alpha, \beta, \lambda$, the bias, the mean squared errors (MSE), and the 95\% confidence interval for each set.
The confidence interval is represented by its lower bound (Low), upper bound (Up), and confidence interval width (CIW) (CIW = Up - Low). Each set of results is organized in a tabular format corresponding to the respective parameter configuration and sample size.

\begin{table}[htbp]
  \centering
  \caption{Parameter Set 1  ($\alpha = 2.0, \beta = 2.0, \lambda = 2.0$)}
  \label{table_sim1}
  \vspace{1em}
  \begin{tabular}{
      S[table-format=3.0]  
      c 
      S[table-format=1.2, round-mode=places, round-precision=2]  
      S[table-format=-1.2, round-mode=places, round-precision=2]  
      S[table-format=2.3, round-mode=places, round-precision=3]  
      S[table-format=3.3, round-mode=places, round-precision=3]  
      S[table-format=-2.2, round-mode=places, round-precision=2]  
      S[table-format=2.2, round-mode=places, round-precision=2]  
      S[table-format=2.2, round-mode=places, round-precision=2]  
      @{\hspace{0.5em}}
    }
    \toprule
    {$n$} &           & {MLE}  & {Bias}  & {SD}    & {MSE}      & {Low}   & {Up}   & {CIW}  \\
    \specialrule{0.1pt}{4pt}{4pt}
    100   & $\alpha$  & 2.6488 & 0.6488  & 1.1007  & 1.632244   & 0.491   & 4.806  & 4.315  \\
          & $\beta$   & 2.4092 & 0.4092  & 2.0918  & 4.542053   & -1.691  & 6.509  & 8.200  \\
          & $\lambda$ & 4.0689 & 2.0689  & 18.2388 & 336.866815 & -31.678 & 39.816 & 71.494 \\
    \specialrule{0.1pt}{4pt}{4pt}
    200   & $\alpha$  & 2.5469 & 0.5469  & 1.0421  & 1.384919   & 0.504   & 4.589  & 4.085  \\
          & $\beta$   & 2.0675 & 0.0675  & 1.5071  & 2.275373   & -0.886  & 5.021  & 5.907  \\
          & $\lambda$ & 2.1451 & 0.1451  & 5.8166  & 33.847290  & -9.255  & 13.545 & 22.800 \\
    \specialrule{0.1pt}{4pt}{4pt}
    300   & $\alpha$  & 2.4934 & 0.4934  & 1.0252  & 1.294325   & 0.484   & 4.503  & 4.019  \\
          & $\beta$   & 1.9348 & -0.0652 & 1.2478  & 1.560853   & -0.511  & 4.380  & 4.891  \\
          & $\lambda$ & 1.9618 & -0.0382 & 1.9201  & 3.687378   & -1.801  & 5.725  & 7.526  \\
    \bottomrule
  \end{tabular}
\end{table}

\begin{table}[htbp]
  \centering
  \caption{Parameter Set 2 ($\alpha = 2.0, \beta = 2.0, \lambda = 0.5$)}
  \label{table_sim2}
  \vspace{1em}
  \begin{tabular}{
      S[table-format=3.0]  
      c 
      S[table-format=1.2, round-mode=places, round-precision=2]  
      S[table-format=-1.2, round-mode=places, round-precision=2]  
      S[table-format=1.3, round-mode=places, round-precision=3]  
      S[table-format=1.3, round-mode=places, round-precision=3]  
      S[table-format=-1.2, round-mode=places, round-precision=2]  
      S[table-format=1.2, round-mode=places, round-precision=2]  
      S[table-format=1.2, round-mode=places, round-precision=2]  
      @{\hspace{0.5em}}
    }
    \toprule
    {$n$} &           & {MLE}  & {Bias}  & {SD}   & {MSE}    & {Low}  & {Up}  & {CIW} \\
    \specialrule{0.1pt}{4pt}{4pt}
    100   & $\alpha$  & 1.8218 & -0.1782 & 0.6502 & 0.454429 & 0.547  & 3.096 & 2.549 \\
          & $\beta$   & 2.1688 & 0.1688  & 1.1319 & 1.309400 & -0.050 & 4.387 & 4.437 \\
          & $\lambda$ & 0.8319 & 0.3319  & 0.8528 & 0.837270 & -0.840 & 2.503 & 3.343 \\
    \specialrule{0.1pt}{4pt}{4pt}
    200   & $\alpha$  & 1.8629 & -0.1371 & 0.6328 & 0.419189 & 0.623  & 3.103 & 2.480 \\
          & $\beta$   & 1.9898 & -0.0102 & 0.7820 & 0.611504 & 0.457  & 3.523 & 3.066 \\
          & $\lambda$ & 0.8073 & 0.3073  & 0.8589 & 0.831908 & -0.876 & 2.491 & 3.367 \\
    \specialrule{0.1pt}{4pt}{4pt}
    300   & $\alpha$  & 1.8841 & -0.1159 & 0.6206 & 0.398557 & 0.668  & 3.101 & 2.433 \\
          & $\beta$   & 1.9464 & -0.0536 & 0.6614 & 0.440174 & 0.650  & 3.243 & 2.593 \\
          & $\lambda$ & 0.7808 & 0.2808  & 0.8374 & 0.779934 & -0.860 & 2.422 & 3.282 \\
    \bottomrule
  \end{tabular}
\end{table}

\begin{table}[htbp]
  \centering
  \caption{Parameter Set 3 ($\alpha = 2.0, \beta = 0.5, \lambda = 2.0$)}
  \label{table_sim3}
  \vspace{1em}
  \begin{tabular}{
      S[table-format=3.0]  
      c 
      S[table-format=1.2, round-mode=places, round-precision=2]  
      S[table-format=-1.2, round-mode=places, round-precision=2]  
      S[table-format=1.3, round-mode=places, round-precision=3]  
      S[table-format=2.3, round-mode=places, round-precision=3]  
      S[table-format=-1.2, round-mode=places, round-precision=2]  
      S[table-format=2.2, round-mode=places, round-precision=2]  
      S[table-format=2.2, round-mode=places, round-precision=2]  
      @{\hspace{0.5em}}
    }
    \toprule
    {$n$} &           & {MLE}  & {Bias}  & {SD}   & {MSE}     & {Low}  & {Up}   & {CIW}  \\
    \specialrule{0.1pt}{4pt}{4pt}
    100   & $\alpha$  & 2.4055 & 0.4055  & 0.8435 & 0.875812  & 0.752  & 4.059  & 3.307  \\
          & $\beta$   & 1.0513 & 0.5513  & 1.4241 & 2.331532  & -1.740 & 3.842  & 5.582  \\
          & $\lambda$ & 1.9741 & -0.0259 & 5.4596 & 29.802053 & -8.727 & 12.675 & 21.402 \\
    \specialrule{0.1pt}{4pt}{4pt}
    200   & $\alpha$  & 2.3065 & 0.3065  & 0.7917 & 0.720581  & 0.755  & 3.858  & 3.103  \\
          & $\beta$   & 0.6725 & 0.1725  & 0.8297 & 0.718032  & -0.954 & 2.299  & 3.253  \\
          & $\lambda$ & 1.9179 & -0.0821 & 1.1986 & 1.443081  & -0.431 & 4.267  & 4.698  \\
    \specialrule{0.1pt}{4pt}{4pt}
    300   & $\alpha$  & 2.2574 & 0.2574  & 0.7560 & 0.637628  & 0.776  & 3.739  & 2.963  \\
          & $\beta$   & 0.5618 & 0.0618  & 0.6381 & 0.410952  & -0.689 & 1.813  & 2.502  \\
          & $\lambda$ & 1.9771 & -0.0229 & 1.1946 & 1.427423  & -0.364 & 4.319  & 4.683  \\
    \bottomrule
  \end{tabular}
\end{table}

\begin{table}[htbp]
  \centering
  \caption{Parameter Set 4 ($\alpha = 2.0, \beta = 0.5, \lambda = 0.5$)}
  \label{table_sim4}
  \vspace{1em}
  \begin{tabular}{
      S[table-format=3.0]  
      c 
      S[table-format=1.2, round-mode=places, round-precision=2]  
      S[table-format=-1.2, round-mode=places, round-precision=2]  
      S[table-format=1.3, round-mode=places, round-precision=3]  
      S[table-format=1.3, round-mode=places, round-precision=3]  
      S[table-format=-1.2, round-mode=places, round-precision=2]  
      S[table-format=1.2, round-mode=places, round-precision=2]  
      S[table-format=1.2, round-mode=places, round-precision=2]  
      @{\hspace{0.5em}}
    }
    \toprule
    {$n$} &           & {MLE}  & {Bias}  & {SD}   & {MSE}    & {Low}  & {Up}  & {CIW} \\
    \specialrule{0.1pt}{4pt}{4pt}
    100   & $\alpha$  & 1.8067 & -0.1933 & 0.5332 & 0.321601 & 0.762  & 2.852 & 2.090 \\
          & $\beta$   & 0.7655 & 0.2655  & 0.6557 & 0.500328 & -0.520 & 2.051 & 2.571 \\
          & $\lambda$ & 0.7634 & 0.2634  & 0.6501 & 0.491874 & -0.511 & 2.037 & 2.548 \\
    \specialrule{0.1pt}{4pt}{4pt}
    200   & $\alpha$  & 1.8653 & -0.1347 & 0.5189 & 0.287309 & 0.848  & 2.882 & 2.034 \\
          & $\beta$   & 0.6204 & 0.1204  & 0.4613 & 0.227306 & -0.284 & 1.525 & 1.809 \\
          & $\lambda$ & 0.7101 & 0.2101  & 0.5623 & 0.360279 & -0.392 & 1.812 & 2.204 \\
    \specialrule{0.1pt}{4pt}{4pt}
    300   & $\alpha$  & 1.8918 & -0.1082 & 0.5149 & 0.276805 & 0.883  & 2.901 & 2.018 \\
          & $\beta$   & 0.5754 & 0.0754  & 0.4058 & 0.170291 & -0.220 & 1.371 & 1.591 \\
          & $\lambda$ & 0.6828 & 0.1828  & 0.5303 & 0.314588 & -0.357 & 1.722 & 2.079 \\
    \bottomrule
  \end{tabular}
\end{table}

\begin{table}[htbp]
  \centering
  \caption{Parameter Set 5 ($\alpha = 0.5, \beta = 2.0, \lambda = 2.0$)}
  \label{table_sim5}
  \vspace{1em}
  \begin{tabular}{
      S[table-format=3.0]  
      c 
      S[table-format=1.2, round-mode=places, round-precision=2]  
      S[table-format=-1.2, round-mode=places, round-precision=2]  
      S[table-format=2.3, round-mode=places, round-precision=3]  
      S[table-format=3.3, round-mode=places, round-precision=3]  
      S[table-format=-2.2, round-mode=places, round-precision=2]  
      S[table-format=2.2, round-mode=places, round-precision=2]  
      S[table-format=2.2, round-mode=places, round-precision=2]  
      @{\hspace{0.5em}}
    }
    \toprule
    {$n$} &           & {MLE}  & {Bias} & {SD}    & {MSE}      & {Low}   & {Up}   & {CIW}  \\
    \specialrule{0.1pt}{4pt}{4pt}
    100   & $\alpha$  & 0.8145 & 0.3145 & 0.5289  & 0.378608   & -0.222  & 1.851  & 2.073  \\
          & $\beta$   & 2.4510 & 0.4510 & 1.1566  & 1.540837   & 0.184   & 4.718  & 4.534  \\
          & $\lambda$ & 2.9627 & 0.9627 & 11.3550 & 129.838123 & -19.293 & 25.218 & 44.511 \\
    \specialrule{0.1pt}{4pt}{4pt}
    200   & $\alpha$  & 0.7181 & 0.2181 & 0.4453  & 0.245863   & -0.155  & 1.591  & 1.746  \\
          & $\beta$   & 2.2906 & 0.2906 & 0.9448  & 0.976911   & 0.439   & 4.142  & 3.703  \\
          & $\lambda$ & 2.0510 & 0.0510 & 3.7559  & 14.106312  & -5.310  & 9.412  & 14.722 \\
    \specialrule{0.1pt}{4pt}{4pt}
    300   & $\alpha$  & 0.6613 & 0.1613 & 0.3817  & 0.171692   & -0.087  & 1.409  & 1.496  \\
          & $\beta$   & 2.2196 & 0.2196 & 0.8507  & 0.771732   & 0.552   & 3.887  & 3.335  \\
          & $\lambda$ & 2.0132 & 0.0132 & 1.9913  & 3.964674   & -1.890  & 5.916  & 7.806  \\
    \bottomrule
  \end{tabular}
\end{table}

\begin{table}[htbp]
  \centering
  \caption{Parameter Set 6 ($\alpha = 0.5, \beta = 2.0, \lambda = 0.5$)}
  \label{table_sim6}
  \vspace{1em}
  \begin{tabular}{
      S[table-format=3.0]  
      c 
      S[table-format=1.2, round-mode=places, round-precision=2]  
      S[table-format=-1.2, round-mode=places, round-precision=2]  
      S[table-format=1.3, round-mode=places, round-precision=3]  
      S[table-format=1.3, round-mode=places, round-precision=3]  
      S[table-format=-1.2, round-mode=places, round-precision=2]  
      S[table-format=1.2, round-mode=places, round-precision=2]  
      S[table-format=1.2, round-mode=places, round-precision=2]  
      @{\hspace{0.5em}}
    }
    \toprule
    {$n$} &           & {MLE}  & {Bias}  & {SD}   & {MSE}    & {Low}  & {Up}  & {CIW} \\
    \specialrule{0.1pt}{4pt}{4pt}
    100   & $\alpha$  & 0.5499 & 0.0499  & 0.3011 & 0.093114 & -0.040 & 1.140 & 1.180 \\
          & $\beta$   & 1.9479 & -0.0521 & 0.6273 & 0.396203 & 0.718  & 3.177 & 2.459 \\
          & $\lambda$ & 0.7059 & 0.2059  & 1.0606 & 1.167055 & -1.373 & 2.785 & 4.158 \\
    \specialrule{0.1pt}{4pt}{4pt}
    200   & $\alpha$  & 0.5522 & 0.0522  & 0.2616 & 0.071127 & 0.039  & 1.065 & 1.026 \\
          & $\beta$   & 1.9245 & -0.0755 & 0.4717 & 0.228112 & 1.000  & 2.849 & 1.849 \\
          & $\lambda$ & 0.6611 & 0.1611  & 0.9638 & 0.954756 & -1.228 & 2.550 & 3.778 \\
    \specialrule{0.1pt}{4pt}{4pt}
    300   & $\alpha$  & 0.5514 & 0.0514  & 0.2374 & 0.059001 & 0.086  & 1.017 & 0.931 \\
          & $\beta$   & 1.9366 & -0.0634 & 0.4015 & 0.165169 & 1.150  & 2.724 & 1.574 \\
          & $\lambda$ & 0.6152 & 0.1152  & 0.8686 & 0.767653 & -1.087 & 2.318 & 3.405 \\
    \bottomrule
  \end{tabular}
\end{table}

\begin{table}[htbp]
  \centering
  \caption{Parameter Set 7 ($\alpha = 0.5, \beta = 0.5, \lambda = 2.0$)}
  \label{table_sim7}
  \vspace{1em}
  \begin{tabular}{
      S[table-format=3.0]  
      c 
      S[table-format=1.2, round-mode=places, round-precision=2]  
      S[table-format=-1.2, round-mode=places, round-precision=2]  
      S[table-format=1.3, round-mode=places, round-precision=3]  
      S[table-format=1.3, round-mode=places, round-precision=3]  
      S[table-format=-1.2, round-mode=places, round-precision=2]  
      S[table-format=1.2, round-mode=places, round-precision=2]  
      S[table-format=2.2, round-mode=places, round-precision=2]  
      @{\hspace{0.5em}}
    }
    \toprule
    {$n$} &           & {MLE}  & {Bias}  & {SD}   & {MSE}    & {Low}  & {Up}  & {CIW} \\
    \specialrule{0.1pt}{4pt}{4pt}
    100   & $\alpha$  & 0.7690 & 0.2690  & 0.3768 & 0.214289 & 0.030  & 1.507 & 1.477 \\
          & $\beta$   & 0.6194 & 0.1194  & 0.3667 & 0.148674 & -0.099 & 1.338 & 1.437 \\
          & $\lambda$ & 1.5414 & -0.4586 & 1.6865 & 3.053987 & -1.764 & 4.847 & 6.611 \\
    \specialrule{0.1pt}{4pt}{4pt}
    200   & $\alpha$  & 0.7044 & 0.2044  & 0.3626 & 0.173204 & -0.006 & 1.415 & 1.421 \\
          & $\beta$   & 0.5431 & 0.0431  & 0.2566 & 0.067677 & 0.040  & 1.046 & 1.006 \\
          & $\lambda$ & 1.8215 & -0.1785 & 1.7179 & 2.982594 & -1.546 & 5.189 & 6.735 \\
    \specialrule{0.1pt}{4pt}{4pt}
    300   & $\alpha$  & 0.6645 & 0.1645  & 0.3417 & 0.143774 & -0.005 & 1.334 & 1.339 \\
          & $\beta$   & 0.5227 & 0.0227  & 0.2213 & 0.049486 & 0.089  & 0.956 & 0.867 \\
          & $\lambda$ & 1.9506 & -0.0494 & 1.7405 & 3.031134 & -1.461 & 5.362 & 6.823 \\
    \bottomrule
  \end{tabular}
\end{table}

\begin{table}[htbp]
  \centering
  \caption{Parameter Set 8 ($\alpha = 0.5, \beta = 0.5, \lambda = 0.5$)}
  \label{table_sim8}
  \vspace{1em}
  \begin{tabular}{
      S[table-format=3.0]  
      c  
      S[table-format=1.2, round-mode=places, round-precision=2]  
      S[table-format=-1.2, round-mode=places, round-precision=2]  
      S[table-format=1.3, round-mode=places, round-precision=3]  
      S[table-format=1.3, round-mode=places, round-precision=3]  
      S[table-format=-1.2, round-mode=places, round-precision=2]  
      S[table-format=1.2, round-mode=places, round-precision=2]  
      S[table-format=1.2, round-mode=places, round-precision=2]  
      @{\hspace{0.5em}}
    }
    \toprule
    {$n$} &           & {MLE}  & {Bias}  & {SD}   & {MSE}    & {Low}  & {Up}  & {CIW} \\
    \specialrule{0.1pt}{4pt}{4pt}
    100   & $\alpha$  & 0.5121 & 0.0121  & 0.2283 & 0.052261 & 0.065  & 0.960 & 0.895 \\
          & $\beta$   & 0.4825 & -0.0175 & 0.1878 & 0.035557 & 0.115  & 0.851 & 0.736 \\
          & $\lambda$ & 0.7693 & 0.2693  & 1.1003 & 1.282892 & -1.387 & 2.926 & 4.313 \\
    \specialrule{0.1pt}{4pt}{4pt}
    200   & $\alpha$  & 0.5239 & 0.0239  & 0.2136 & 0.046199 & 0.105  & 0.943 & 0.838 \\
          & $\beta$   & 0.4671 & -0.0329 & 0.1343 & 0.019129 & 0.204  & 0.730 & 0.526 \\
          & $\lambda$ & 0.7148 & 0.2148  & 1.0187 & 1.083656 & -1.282 & 2.711 & 3.993 \\
    \specialrule{0.1pt}{4pt}{4pt}
    300   & $\alpha$  & 0.5294 & 0.0294  & 0.2027 & 0.041930 & 0.132  & 0.927 & 0.795 \\
          & $\beta$   & 0.4676 & -0.0324 & 0.1141 & 0.014061 & 0.244  & 0.691 & 0.447 \\
          & $\lambda$ & 0.6688 & 0.1688  & 0.9658 & 0.961113 & -1.224 & 2.562 & 3.786 \\
    \bottomrule
    \addlinespace[2em]
  \end{tabular}
\end{table}

\FloatBarrier

We observed that the bias, SD, MSE and CIW of the estimators $\hat{\alpha}, \hat{\beta}, \hat{\lambda}$, obtained via MLE, tend to decrease as the sample size increases
from $n = 100$ to $n = 300$.

\section{Data Analysis}\label{sec_data}

This section presents the empirical evaluation of the proposed CLFRD model against existing lifetime models using real data. For each dataset, we compare the proposed CLFRD model against well-known models like LFRD, RD, ED and GED \cite{Gupta1999}. For each model, the MLEs and the corresponding negative log-likelihood value $(\mathcal{L})$ were computed. Goodness-of-fit was further evaluated using the Kolmogorov–Smirnov (K-S) test. Anderson-Darling test (AD), and Cramér–von Mises test (CM)  statistics, and the AIC values for the models were also calculated, with the CLFRD model treated as the baseline. We derived the variance-covariance matrix. We also plotted the empirical survival function against the fitted survival functions for visualising the performance of the models.

\subsection{Sample Data 1}

The dataset consists of the marks achieved by 48 students in Mathematics \cite{Gupta2009}. \autoref{table_dataset1} shows the dataset and \autoref{table_mle1} shows the computed MLEs on the dataset.

\begin{table}[H]
  \centering
  \caption{Marks Obtained by 48 Students in Mathematics}
  \label{table_dataset1}
  \vspace{1em}
  \begin{tabular}{*{12}{c}}
    \toprule
    4  & 5  & 6  & 6  & 7  & 7  & 8  & 11 & 12 & 12 & 13 & 14 \\
    14 & 15 & 15 & 15 & 15 & 15 & 18 & 18 & 18 & 19 & 19 & 19 \\
    20 & 21 & 21 & 23 & 23 & 23 & 25 & 27 & 28 & 29 & 31 & 34 \\
    34 & 37 & 39 & 40 & 44 & 50 & 50 & 58 & 60 & 65 & 70 & 86 \\
    \bottomrule
  \end{tabular}
\end{table}

\begin{table}[H]
  \centering
  \caption{Maximum Likelihood Estimates of Model Parameters}
  \label{table_mle1}
  \vspace{1em}
  \begin{tabular}{lp{8cm}}
    \toprule
    Model & MLEs                                                                              \\
    \midrule
    LFRD  & $\alpha = 1.55 \times 10^{-2}$, $\beta = 1.19 \times 10^{-3}$                     \\
    \addlinespace[0.2em]
    RD    & $\sigma = 22.4669$                                                                \\
    \addlinespace[0.2em]
    ED    & $\lambda = 3.86 \times 10^{-2}$                                                   \\
    \addlinespace[0.2em]
    GED   & $\lambda = 6.55 \times 10^{-2}$, $\alpha = 2.5212$                                \\
    \addlinespace[0.2em]
    CLFRD & $\alpha = 6.19 \times 10^{-4}$, $\beta = 1.02 \times 10^{-3}$, $\lambda = 1.7140$ \\
    \bottomrule
  \end{tabular}
\end{table}

As shown in \autoref{table_modelcomp1}, the log-likelihood value of the CLFRD model is higher than that of the other models with the exception of GED. The K-S test also indicates that the CLFRD and GED models are the best fit models and can capture the data's structure well. They also have the highest p-value. The AIC values show that the CLFRD model is very similar to the LFRD, GRD and GED models.

\begin{table}[H]
  \centering
  \caption{Model Comparison Results}
  \label{table_modelcomp1}
  \vspace{1em}
  \begin{tabular}{lcccccc}
    \toprule
    Model & $-2\mathcal{L}$ & K-S    & $p$-value & AIC     & AD     & CM     \\
    \midrule
    LFRD  & 400.32          & 0.1365 & 0.3326    & 402.321 & 0.9246 & 0.1615 \\
    \addlinespace[0.2em]
    RD    & 404.30          & 0.2171 & 0.0216    & 404.309 & 3.0924 & 0.5876 \\
    \addlinespace[0.2em]
    ED    & 408.40          & 0.2042 & 0.0365    & 408.392 & 2.5876 & 0.4469 \\
    \addlinespace[0.2em]
    GED   & 393.62          & 0.0937 & 0.7935    & 395.616 & 0.3521 & 0.0594 \\
    \addlinespace[0.2em]
    CLFRD & 396.10          & 0.1190 & 0.5048    & 400.108 & 0.7048 & 0.1197 \\
    \bottomrule
  \end{tabular}
\end{table}

The variance-covariance matrix for the CLFRD parameters is given as:

\begin{align*}
  I_0^{-1} =
  \begin{bmatrix}
    1.19 \times 10^{-5}  & -5.25 \times 10^{-7} & -7.94 \times 10^{-4} \\
    -5.25 \times 10^{-7} & 2.20 \times 10^{-7}  & -4.09 \times 10^{-4} \\
    -7.94 \times 10^{-4} & -4.09 \times 10^{-4} & 1.2705
  \end{bmatrix}
  .
\end{align*}

\FloatBarrier

\subsection{Sample Data 2}

The dataset consists of the failure times for 36 appliances subjected to an automatic life test \cite{Lawless1982}. We divided each value by 1000 for ease of computation. \autoref{table_dataset2} shows the dataset and \autoref{table_mle2} shows the computed MLEs on the dataset.

\begin{table}[H]
  \centering
  \caption{Failure Time of 36 appliances}
  \label{table_dataset2}
  \vspace{1em}
  \begin{tabular}{*{9}{c}}
    \toprule
    11   & 35   & 49   & 170  & 329  & 381  & 708  & 958  & 1062  \\
    1167 & 1594 & 1925 & 1990 & 2223 & 2327 & 2400 & 2451 & 2471  \\
    2551 & 2565 & 2568 & 2694 & 2702 & 2761 & 2831 & 3034 & 3059  \\
    3112 & 3214 & 3478 & 3504 & 4329 & 6367 & 6976 & 7846 & 13403 \\
    \bottomrule
  \end{tabular}
\end{table}

\begin{table}[H]
  \centering
  \caption{Maximum Likelihood Estimates of Model Parameters}
  \label{table_mle2}
  \vspace{1em}
  \begin{tabular}{lp{8cm}}
    \toprule
    Model & MLEs                                                                              \\
    \midrule
    LFRD  & $\alpha = 0.3254$ , $\beta = 1.47 \times 10^{-2}$                                 \\
    \addlinespace[0.2em]
    RD    & $\sigma = 2.6473$                                                                 \\
    \addlinespace[0.2em]
    ED    & $\lambda = 0.3627$                                                                \\
    \addlinespace[0.2em]
    GED   & $\lambda = 0.3535$, $\alpha = 0.9603$                                             \\
    \addlinespace[0.2em]
    CLFRD & $\alpha = 6.38 \times 10^{-2}$, $\beta = 2.58 \times 10^{-2}$, $\lambda = 2.7986$ \\
    \bottomrule
  \end{tabular}
\end{table}

As shown in \autoref{table_modelcomp2}, the log-likelihood value of the CLFRD model is higher than that of the other models. The base LFRD model is similar to our proposed model. The K-S test also indicates that the CLFRD model is a good fit model and can capture the data's structure well. The AIC values show that the CLFRD model is similar to LFRD.

\begin{table}[H]
  \centering
  \caption{Model Comparison Results}
  \label{table_modelcomp2}
  \vspace{1em}
  \begin{tabular}{lcccccc}
    \toprule
    Model & $-2\mathcal{L}$ & K-S    & $p$-value & AIC     & AD     & CM     \\
    \midrule
    LFRD  & 144.72          & 0.1743 & 0.1993    & 146.719 & 1.3549 & 0.2536 \\
    \addlinespace[0.2em]
    RD    & 182.12          & 0.2841 & 0.0046    & 182.117 & 7.6518 & 0.9272 \\
    \addlinespace[0.2em]
    ED    & 145.02          & 0.1970 & 0.1064    & 145.013 & 1.4970 & 0.2995 \\
    \addlinespace[0.2em]
    GED   & 144.98          & 0.2021 & 0.0914    & 146.977 & 1.5298 & 0.3145 \\
    \addlinespace[0.2em]
    CLFRD & 143.28          & 0.1551 & 0.3181    & 147.285 & 1.2365 & 0.2017 \\
    \bottomrule
  \end{tabular}
\end{table}

The variance-covariance matrix for the CLFRD parameters is given as:

\begin{align*}
  I_0^{-1} =
  \begin{bmatrix}
    2.25 \times 10^{-3}  & 4.27 \times 10^{-4}  & -9.27 \times 10^{-2} \\
    4.27 \times 10^{-4}  & 4.98 \times 10^{-4}  & -3.96 \times 10^{-2} \\
    -9.27 \times 10^{-2} & -3.96 \times 10^{-2} & 5.4571
  \end{bmatrix}
  .
\end{align*}

\FloatBarrier

\subsection{Sample Data 3}

The dataset consists of the lifetime of 50 devices \cite{Aarset1987}. \autoref{table_dataset3} shows the dataset and \autoref{table_mle3} shows the computed MLEs on the dataset.

\begin{table}[H]
  \centering
  \caption{Lifetime of 50 Devices}
  \label{table_dataset3}
  \vspace{1em}
  \begin{tabular}{*{10}{c}}
    \toprule
    0.1 & 0.2 & 1  & 1  & 1  & 1  & 1  & 2  & 3  & 6  \\
    7   & 11  & 12 & 18 & 18 & 18 & 18 & 18 & 21 & 32 \\
    36  & 40  & 45 & 46 & 47 & 50 & 55 & 60 & 63 & 63 \\
    67  & 67  & 67 & 67 & 72 & 75 & 79 & 82 & 82 & 83 \\
    84  & 84  & 84 & 85 & 85 & 85 & 85 & 85 & 86 & 86 \\
    \bottomrule
  \end{tabular}
\end{table}

\begin{table}[!htbp]
  \centering
  \caption{Maximum Likelihood Estimates of Model Parameters}
  \label{table_mle3}
  \vspace{1em}
  \begin{tabular}{lp{9cm}}
    \toprule
    Model & MLEs                                                                                           \\
    \midrule
    LFRD  & $\alpha = 1.36 \times 10^{-2}$, $\beta = 2.40 \times 10^{-4}$                                  \\
    \addlinespace[0.2em]
    RD    & $\sigma = 39.6472$                                                                             \\
    \addlinespace[0.2em]
    ED    & $\lambda = 2.19 \times 10^{-2}$                                                                \\
    \addlinespace[0.2em]
    GED   & $\lambda = 1.87 \times 10^{-2}$, $\alpha = 0.7802$                                             \\
    \addlinespace[0.2em]
    CLFRD & $\alpha = 1.60 \times 10^{-2}$, $\beta = 1.57 \times 10^{-4}$, $\lambda = 5.57 \times 10^{-3}$ \\
    \bottomrule
  \end{tabular}
\end{table}

As shown in \autoref{table_modelcomp3}, the log-likelihood value of the CLFRD model is higher than that of the other models. The base LFRD model is similar to our proposed model. The K-S test also indicates that the CLFRD model is the best fit model and can capture the data's structure well, having the highest p-value. The AIC values show that CLFRD is not worse than ED and RD and is very similar to the base LFRD.

\begin{table}[H]
  \centering
  \caption{Model Comparison Results}
  \label{table_modelcomp3}
  \vspace{1em}
  \begin{tabular}{lcccccc}
    \toprule
    Model & $-2\mathcal{L}$ & K-S    & $p$-value & AIC     & AD      & CM     \\
    \midrule
    LFRD  & 476.12          & 0.1769 & 0.0876    & 478.127 & 4.0346  & 0.4627 \\
    \addlinespace[0.2em]
    RD    & 528.10          & 0.2621 & 0.0021    & 528.106 & 13.3206 & 0.7913 \\
    \addlinespace[0.2em]
    ED    & 482.18          & 0.1913 & 0.0515    & 482.179 & 3.6542  & 0.5199 \\
    \addlinespace[0.2em]
    GED   & 480.00          & 0.2044 & 0.0307    & 481.990 & 3.2585  & 0.5667 \\
    \addlinespace[0.2em]
    CLFRD & 476.84          & 0.1744 & 0.0956    & 480.829 & 3.6818  & 0.4420 \\
    \bottomrule
  \end{tabular}
\end{table}

The variance-covariance matrix for the CLFRD parameters is given as:

\begin{align*}
  I_0^{-1} =
  \begin{bmatrix}
    5.237884 \times 10^{-5}  & -3.919999 \times 10^{-7} & -2.867154 \times 10^{-3} \\
    -3.919999 \times 10^{-7} & 9.795843 \times 10^{-9}  & 8.299924 \times 10^{-6}  \\
    -2.867154 \times 10^{-3} & 8.299924 \times 10^{-6}  & 2.409978 \times 10^{-1}
  \end{bmatrix}
  .
\end{align*}

Finally, the empirical and the fitted survival functions for each model are visually compared in \autoref{fig_surv}. We can see that the CLFRD model closely follows the empirical survival function for all datasets.

\begin{figure}[!htbp]
  \vspace{3em}
  \centering
  \begin{minipage}{0.32\textwidth}
    \centering
    \includegraphics[width=\textwidth]{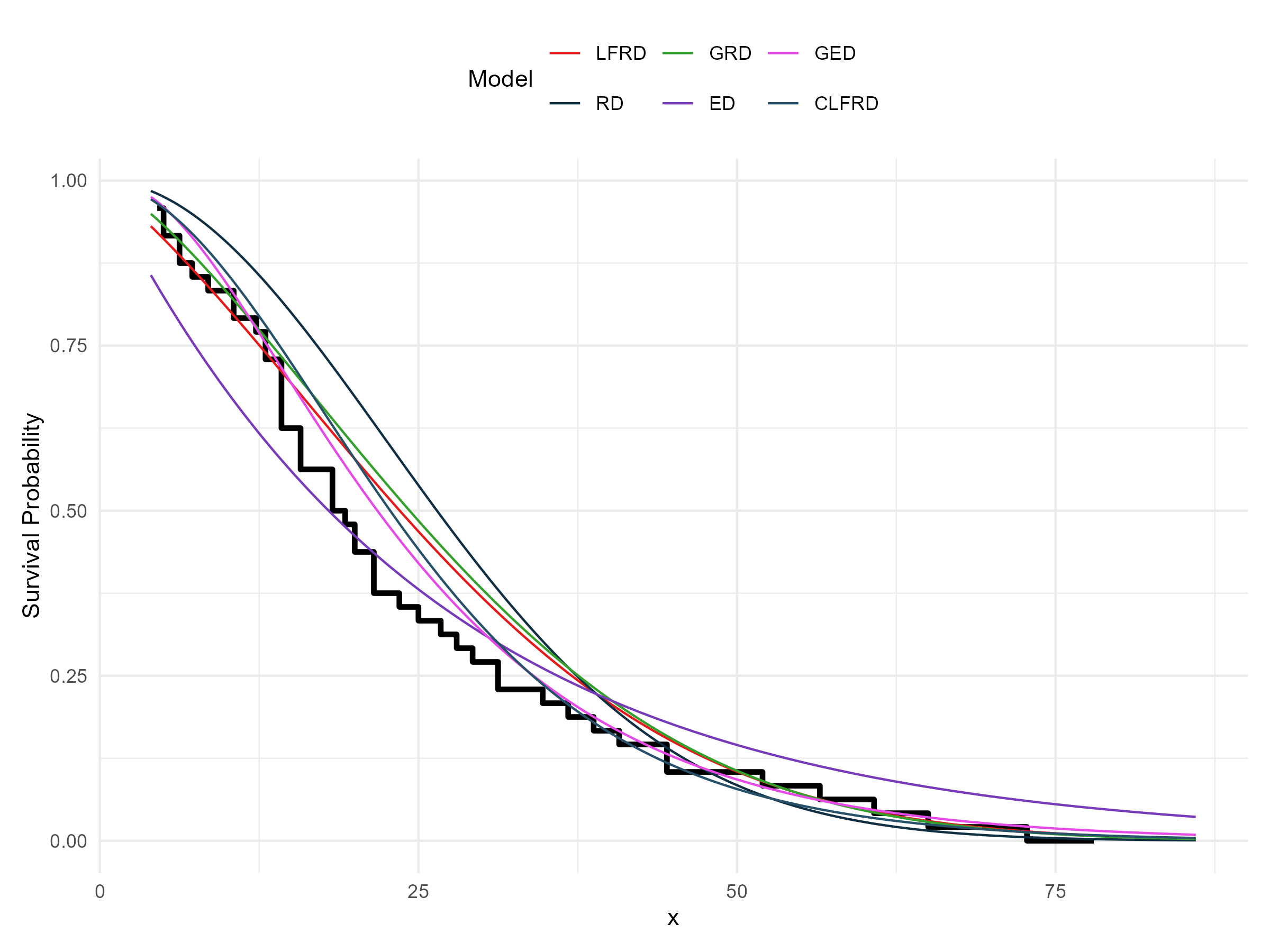}
    {\vspace{0.5ex}\centering\scriptsize Sample Data 1\par}
  \end{minipage}
  \hfill
  \begin{minipage}{0.32\textwidth}
    \centering
    \includegraphics[width=\textwidth]{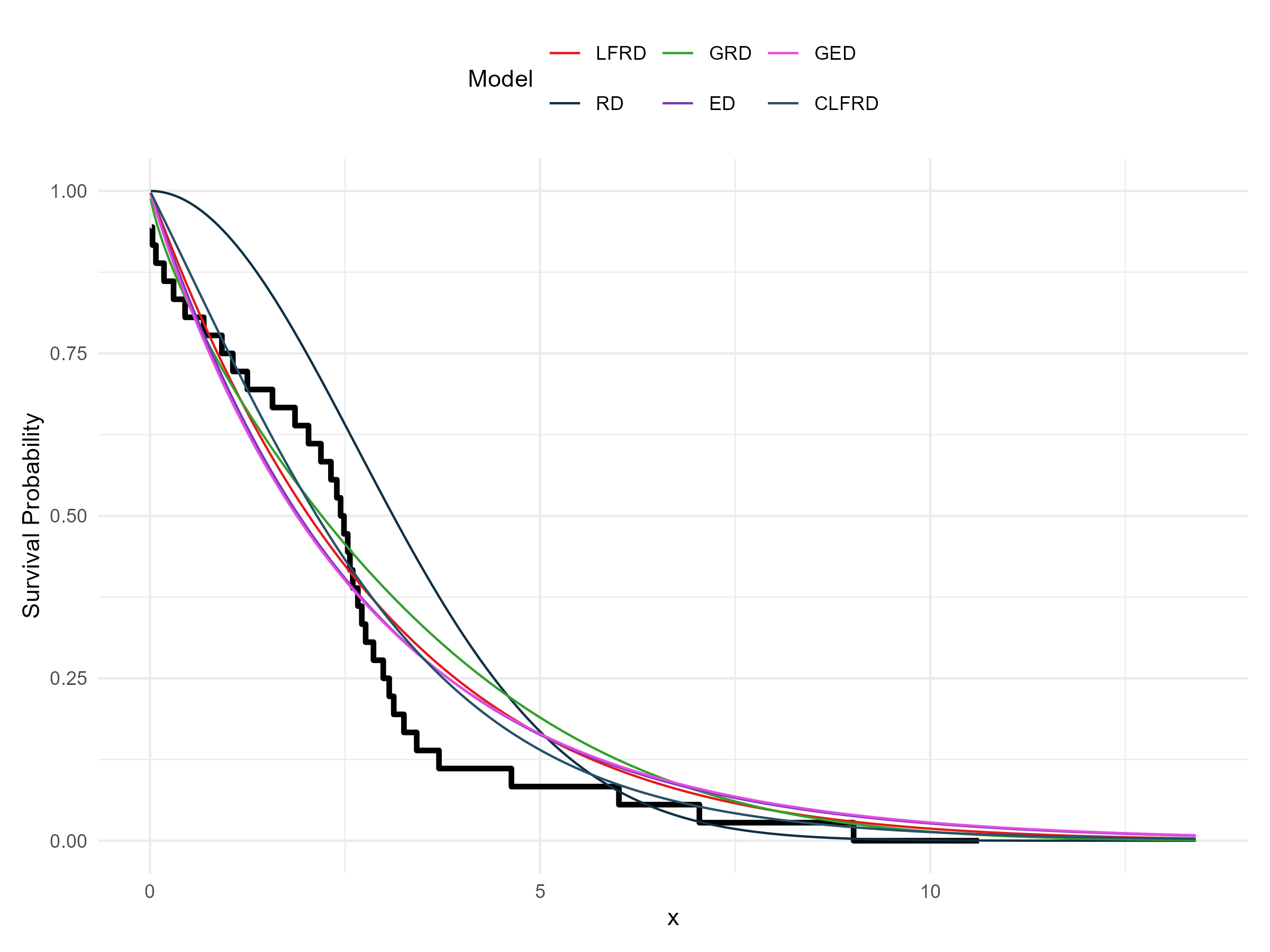}
    {\vspace{0.5ex}\centering\scriptsize Sample Data 2\par}
  \end{minipage}
  \hfill
  \begin{minipage}{0.32\textwidth}
    \centering
    \includegraphics[width=\textwidth]{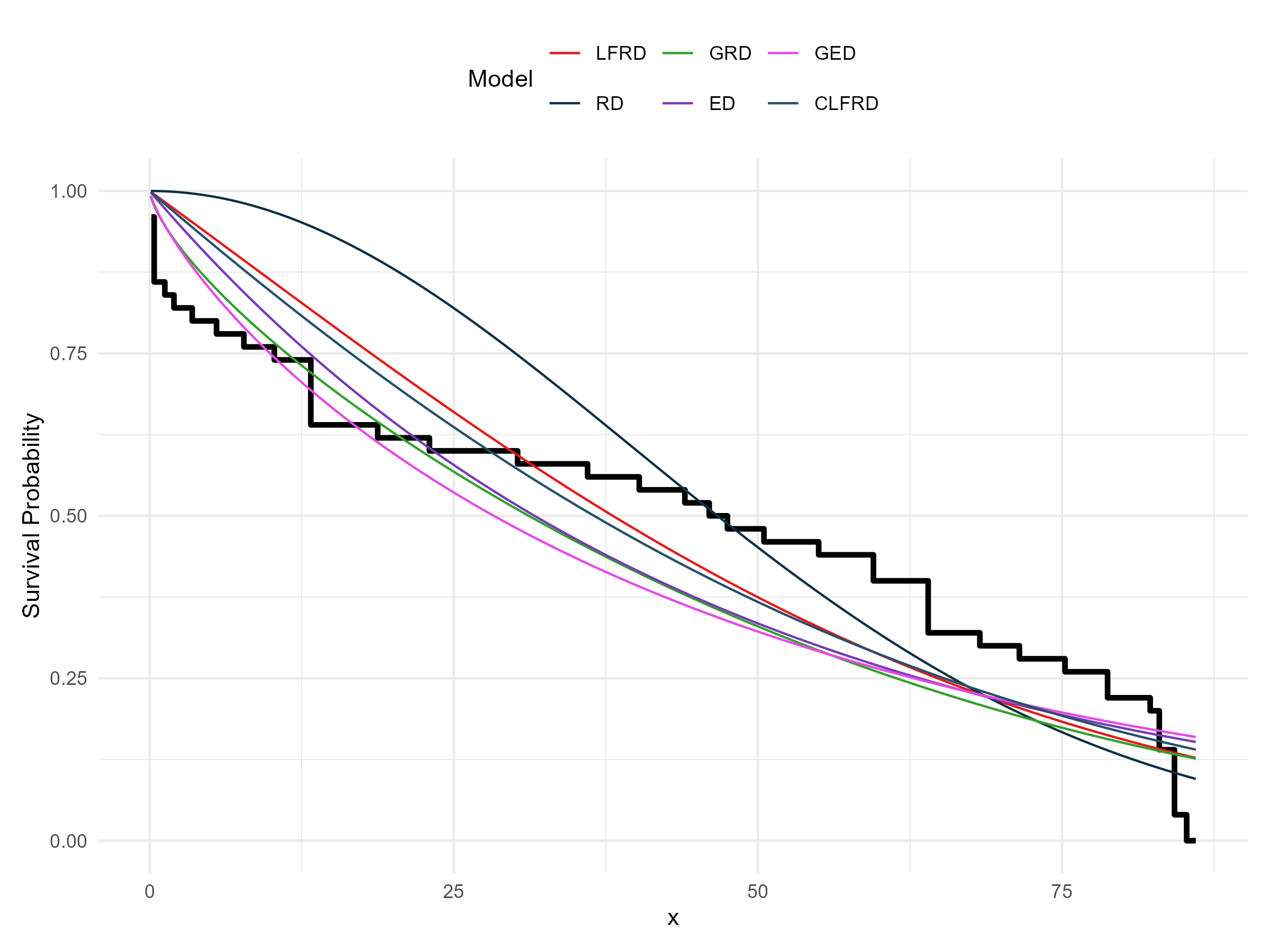}
    {\vspace{0.5ex}\centering\scriptsize Sample Data 3\par}
  \end{minipage}
  \caption{Survival Functions Fitted vs. Empirical}
  \label{fig_surv}
\end{figure}

\FloatBarrier

\section{Conclusion}\label{sec_con}

This paper has presented a comprehensive theoretical, statistical, and computational analysis of the new proposed Compounded Linear Failure Rate Distribution (CLFRD), a three-parameter extension of the classical Linear Failure Rate model. Our study established the CLFRD's mathematical foundations by deriving its key properties, including the probability density function, hazard rate function, reversed hazard rate function, mean residual life, and mean inactivity time, while demonstrating its ability to generalize both the Exponential and the Rayleigh distributions through parameter reduction. We further characterized the model's behavior through moment calculations, quantile analysis, and order statistics. The development of complete likelihood-based inference procedures, including the observed information matrix, enables robust parameter estimation with computable variance-covariance structures. Through a detailed systematic simulation study and comparative empirical analyses, we have validated the CLFRD's superior flexibility in modeling complex failure patterns while maintaining practical implementation feasibility. The model's balanced combination of theoretical generality and computational tractability suggests promising applications in reliability engineering and survival analysis, particularly for systems exhibiting non-standard hazard profiles.

\bibliographystyle{elsarticle-num}
\bibliography{references}

\end{document}